\documentclass[11pt,letterpaper]{amsart}
\usepackage{amsmath,amssymb,amsthm,enumerate,gnuplot-lua-tikz,placeins}
\DeclareMathOperator{\DF}{DF}
\DeclareMathOperator{\MF}{MF}
\DeclareMathOperator{\CDF}{CDF}
\DeclareMathOperator{\CMF}{CMF}

\DeclareMathOperator{\Tr}{Tr}
\DeclareMathOperator{\app}{app}
\DeclareMathOperator{\trunc}{trunc}
\renewcommand{\phi}{\varphi}

\newcommand{\C}{{\mathbb C}}
\newcommand{\F}{{\mathbb F}}
\newcommand{\R}{{\mathbb R}}
\newcommand{\Rplus}{\R_+}
\newcommand{\Z}{{\mathbb Z}}
\newcommand{\Fp}{\F_p}
\newcommand{\Fpu}{\F_p^*}
\newcommand{\Fpn}{\F_{p^n}}
\newcommand{\Fpnu}{\F_{p^n}^*}
\newcommand{\Fu}{F^*}
\newcommand{\Fi}{F_\iota}

\newcommand{\norm}[2]{\Vert{#1}\Vert_{#2}}
\newcommand{\normp}[3]{\norm{#1}{#2}^{#3}}
\newcommand{\normt}[1]{\norm{#1}{2}}

\newcommand{\normtt}[1]{\normp{#1}{2}{2}}
\newcommand{\normtf}[1]{\normp{#1}{2}{4}}
\newcommand{\normff}[1]{\normp{#1}{4}{4}}
\newcommand{\cnr}[2]{\normtt{#1 #2}/(\normtt{#1} \normtt{#2})}
\newcommand{\cnrd}[2]{\frac{\normtt{{#1} {#2}}}{\normtt{#1}  \normtt{#2}}}
\newcommand{\Lambdaapp}{\Lambda_{\app}}
\newcommand{\Lambdatrunc}{\Lambda_{\trunc}}
\newcommand{\achars}{\widehat{F}}
\newcommand{\mchars}{\widehat{\Fu}}
\newcommand{\card}[1]{\left|{#1}\right|}
\newcommand{\conj}[1]{\overline{#1}}
\newcommand{\sums}[1]{\sum_{\substack{#1}}}
\newcommand{\floor}[1]{\lfloor{#1}\rfloor}
\newtheorem{theorem}{Theorem}[section]
\newtheorem{lemma}[theorem]{Lemma}

\title[Crosscorrelation of Sequences Derived from Characters]{Aperiodic Crosscorrelation of Sequences Derived from Characters}
\author{Daniel J.~Katz}
\date{first version: 29 February 2016; this version: 05 July 2016}
\keywords{crosscorrelation, autocorrelation, aperiodic, merit factor, m-sequence, Legendre sequence}
\thanks{Daniel J.~Katz is at California State University, Northridge}
\thanks{This work was supported in part by: NSF Grant DMS 1500856; a Research, Scholarship, and Creative Activity Award from California State University, Northridge; and funding from an NSERC grant awarded to Jonathan Jedwab.}
\begin{document}
\begin{abstract}
It is shown that pairs of maximal linear recursive sequences (m-sequences) typically have mean square aperiodic crosscorrelation on par with that of random sequences, but that if one takes a pair of m-sequences where one is the reverse of the other, and shifts them appropriately, one can get significantly lower mean square aperiodic crosscorrelation.  Sequence pairs with even lower mean square aperiodic crosscorrelation are constructed by taking a Legendre sequence, cyclically shifting it, and then cutting it (approximately) in half and using the halves as the sequences of the pair.  In some of these constructions, the mean square aperiodic crosscorrelation can be lowered further if one truncates or periodically extends (appends) the sequences.  Exact asymptotic formulae for mean squared aperiodic crosscorrelation are proved for sequences derived from additive characters (including m-sequences and modified versions thereof) and multiplicative characters (including Legendre sequences and their relatives).  Data is presented that shows that sequences of modest length have performance that closely approximates the asymptotic formulae.
\end{abstract}
\maketitle
\section{Introduction}
Sequences with low autocorrelation and sequence pairs with low crosscorrelation are used in a great variety of applications, including remote sensing, design of scientific instruments, and communications networks.  See \cite{Golomb,Turyn,Golay-1972,Scholtz-Welch,Sarwate-Pursley,Sarwate} for some foundational works, \cite[Parts V, IX]{Schroeder} for an extensive discussion of applications, and \cite{Golomb-Gong} for an entire monograph on design of sequences with favorable correlation properties.
In this paper we investigate the aperiodic autocorrelation and crosscorrelation of sequences derived from finite field characters.
These are finite sequences whose terms lie in the complex field $\C$, and are typically $m$th roots of unity (when $m=2$ the terms are $\pm 1$ and we obtain {\it binary} sequences).
If $f=(f_0,\ldots,f_{\ell-1})$ and $g=(g_0,\ldots,g_{m-1})$ are sequences of respective lengths $\ell$ and $m$, and $s$ is an integer, then the {\it aperiodic crosscorrelation of $f$ with $g$ at shift $s$} is defined to be
\[
C_{f,g}(s) = \sum_{j \in \Z} f_j \conj{g_{j+s}},
\]
where $\conj{u}$ denotes the complex conjugate of $u \in \C$, and where we use the convention that $f_j=0$ when $j\not\in\{0,1,\ldots,\ell-1\}$ and $g_k=0$ when $k \not \in \{0,1,\ldots,m-1\}$.
This makes $C_{f,g}(s)=0$ whenever $s \leq -\ell$ or $s \geq m$.
For the rest of this paper, we shall only consider the case where the sequences are of equal length, that is, where $\ell=m$.

The {\it aperiodic autocorrelation of $f$ at shift $s$} is just the aperiodic crosscorrelation of $f$ with $f$ at shift $s$, that is, $C_{f,f}(s)$.
If our complex sequence has terms that are {\it unimodular} (complex numbers of unit magnitude), then its aperiodic autocorrelation at shift $0$ is equal to the length of the sequence.

We are interested in studying the mean squared magnitude of autocorrelation and crosscorrelation.
For autocorrelation, Golay \cite{Golay-1972} devised the {\it merit factor} of a sequence
\begin{equation}\label{Amanda}
\MF(f)= \frac{|C_{f,f}(0)|^2}{\sums{s \in \Z \\s\not=0} |C_{f,f}(s)|^2},
\end{equation}
which is large for sequences where the autocorrelation values at nonzero shifts are collectively small.
In this paper we shall call $\MF(f)$ the {\it autocorrelation merit factor} to emphasize that it is a measure of autocorrelation performance, with high merit factor indicating a desirable sequence.
As Sarwate \cite[p.~102]{Sarwate} notes, it is more convenient to study the reciprocal of the autocorrelation merit factor, which we shall call the {\it autocorrelation demerit factor}, and write
\begin{equation}\label{Adam}
\DF(f)= \frac{\sums{s \in \Z \\s\not=0} |C_{f,f}(s)|^2}{|C_{f,f}(0)|^2}.
\end{equation}

In a similar vein, as we would like to keep all crosscorrelation values small when compared to the length of the sequences, we consider the analogous ratio, which we call the {\it crosscorrelation demerit factor},
\begin{equation}\label{Cedric}
\CDF(f,g)=\frac{\sum_{s \in \Z} |C_{f,g}(s)|^2}{|C_{f,f}(0)| \cdot |C_{g,g}(0)|},
\end{equation}
When $f$ and $g$ are sequences of the same length $\ell$ whose terms are unimodular complex numbers, the denominator becomes $\ell^2$.
Since $\CDF(f,g)$ is small for sequence pairs where the crosscorrelation values are collectively small, we see that small crosscorrelation demerit factor indicates a desirable sequence pair.
We define the reciprocal to be the {\it crosscorrelation merit factor},
\begin{equation}\label{Carmine}
\CMF(f,g) =\frac{|C_{f,f}(0)| \cdot |C_{g,g}(0)|}{\sum_{s \in \Z} |C_{f,g}(s)|^2}.
\end{equation}
which is of course large for desirable sequence pairs.
It has been commented by K\"arkk\"ainen \cite[p.~149]{Karkkainen} that for code division multiple access (CDMA) applications, mean square crosscorrelation may be a better measure of crosscorrelation performance than maximum absolute value of crosscorrelation values.

Sarwate \cite[eq.~(13)]{Sarwate} has calculated the expected value of the autocorrelation demerit factor when one averages over all binary sequences of a given length $\ell$.  He found that
\begin{equation}\label{David}
E[\DF(f)] = 1 - \frac{1}{\ell}.
\end{equation}
For crosscorrelation, Sarwate \cite[eq.~(38)]{Sarwate} found that when one averages over all pairs of distinct binary sequences $f$ and $g$, one obtains
\begin{equation}\label{Earl}
E[\CDF(f,g)] = 1.
\end{equation}

Pursley and Sarwate \cite[eqs.~(3) and (4)]{Pursley-Sarwate} proved a useful bound that relates autocorrelation and crosscorrelation demerit factors via the Cauchy-Schwarz inequality.
When translated into the notation of this paper, their bound states that
\begin{equation}\label{Gary}
1-\sqrt{\DF(f)\DF(g)} \leq \CDF(f,g) \leq 1+\sqrt{\DF(f) \DF(g)}.
\end{equation}

One should note that the lower bound means that if both sequences in a pair have very low autocorrelation, then their crosscorrelation demerit factor cannot be much lower than $1$, that is, much better than what one typically obtains for randomly selected sequences.
And conversely, a pair of sequences with very low crosscorrelation cannot also both have very low autocorrelation.
Many of our results are close to the lower bound in \eqref{Gary}, and we shall often observe the resulting tradeoff between good autocorrelation and crosscorrelation.

Much effort has gone into finding sequence families with superior autocorrelation properties.
We consider some infinite families, and their merit factor in the limit as length tends to infinity.
In 1968, Littlewood \cite{Littlewood-1968} proved that sequences derived from the Rudin-Shapiro polynomials have asymptotic merit factor $3$, that is, $\MF(f) \to 3$ as the length of $f$ tends to infinity.
This record was exceeded in 1988, when H\o holdt and Jensen \cite{Hoholdt-Jensen}, inspired by observations of Turyn reported by Golay \cite{Golay-1983}, proved that shifted Legendre sequences have asymptotic merit factor $6$.
Although H\o holdt and Jensen conjectured that $6$ could not be asymptotically exceeded, later observations by Kirilusha and Narayanaswamy \cite{Kirilusha-Narayanaswamy} and Borwein, Choi, and Jedwab \cite{Borwein-Choi-Jedwab} suggested that one could increase the asymptotic merit factor of shifted Legendre sequences by appending to the end a repetition of an initial segment of the original.
The third and current asymptotic merit factor record was established in 2013 by Jedwab, Schmidt, and the author of this paper \cite[Theorem 1.1]{Jedwab-Katz-Schmidt-2013-Littlewood}, who proved that appending does enable one to achieve asymptotic merit factor of $6.342061\ldots$, the largest root of $29 x^3-249 x^2+417 x -27$.

On the other hand, relatively little is known about the more difficult problem of characterizing the crosscorrelation of important sequence families.
Since there is only one Legendre sequence of each length, Legendre sequences are not suitable for use in communications systems where pairs (or larger sets) of sequences are used.
Therefore, one generally uses maximal linear recursive sequences (m-sequences), since one can obtain sets of many having the same length and they are easy to implement through linear feedback shift registers.

This paper provides rigorous calculations of asymptotic crosscorrelation merit factors for m-sequences and other sequences derived from finite field characters in the limit as the sequence length tends to infinity.
It also presents computer calculations that show that this limiting behavior is already closely approximated for sequences of modest length.
In Section \ref{Abigail}, we present our main results on m-sequences with minimal technical details.
We show that one can obtain asymptotic crosscorrelation merit factors in excess of $1$ by using appropriately selected pairs of m-sequences.
This gives families of sequence pairs where both the autocorrelation of the individual sequences and the crosscorrelation of the pairs are asymptotically better than randomly sequences.
In Section \ref{Gilbert}, we then consider what happens in the more general class of {\it additive character sequences}, which we obtain from m-sequences by shortening them (truncation) or periodically extending them (appending).
In Section \ref{Milton} we present our main results on sequences derived from the Legendre symbol and other multiplicative characters.
Here we can get even higher crosscorrelation merit factors than is possible for m-sequences, but typically there is a tradeoff between autocorrelation performance and crosscorrelation performance.  The rest of the paper (Sections \ref{Gail}--\ref{Elizabeth}) is devoted to proving the claims we present in Sections \ref{Abigail}--\ref{Milton}.

We should comment that there are other sequence families derived from finite field characters.  These include the Gold and Kasami sequences (themselves constructed from m-sequences), the Gordon-Mills-Welch sequences (which generalize m-sequences), and Sidel'nikov sequences (derived from multiplicative characters in a different way than the standard multiplicative character sequences discussed in this paper).
These other sequences have more complicated structures than the sequences studied in this paper, and thus exact analyses are expected to be more difficult.
Empirical studies by K\"arkk\"ainnen \cite[p.~148]{Karkkainen} suggest that m-sequences, Gold sequences, and Kasami sequences typically have similar mean square aperiodic crosscorrelation. 

\section{M-Sequences}\label{Abigail}

\subsection{Definition of an M-Sequence}\label{Harold}
To present our findings on crosscorrelation of m-sequences, we first need a terminology that enables us to state our results.
Since aperiodic crosscorrelation behavior sometimes depends on which cyclic shift of an m-sequence is used, it is necessary to define m-sequences in such a way that we can unambiguously indicate which cyclic shift we are talking about.
For $p$ a prime and $n$ a positive integer, the {\it canonical additive character} of the finite field $\Fpn$ of order $p^n$ is the function $\epsilon(x)=\exp(2\pi i \Tr(x)/p)$, where $\Tr \colon \Fpn \to \Fp$ is the absolute trace from Galois theory: $\Tr(x)=x+x^p+\cdots+x^{p^{n-1}}$.
For $\alpha$ a primitive element of $\Fpn$, the {\it $p$-ary Galois sequence associated to $\alpha$} (also called the {\it $p$-ary m-sequence of natural shift associated to $\alpha$}) is 
\[
(\epsilon(\alpha^0),\epsilon(\alpha^1),\ldots,\epsilon(\alpha^{p^n-2})).
\]
Now we define the {\it $p$-ary m-sequence associated to $\alpha$ with shift $s$} to be
\[
(\epsilon(\alpha^s),\epsilon(\alpha^{s+1}),\ldots,\epsilon(\alpha^{s+p^n-2})),
\]
which is just what one obtains by taking the Galois sequence associated to $\alpha$ and cyclically shifting it $s$ places to the left.
We usually write ``binary'' rather than ``$2$-ary'' and ``ternary'' rather than ``$3$-ary.''
M-sequences are examples of {\it additive character sequences}, since one obtains them from the values of additive characters.\footnote{We have always used the canonical additive character in our definitions, but if $\epsilon \colon \Fpn \to \C$ is the canonical additive character, then any other nontrivial additive character $\psi\colon \Fpn \to \C$ has the form $\psi(x)=\epsilon(c x)$ for some $c \in \Fpnu$.  If $\alpha$ is a primitive element of $\Fpn$, we write $c=\alpha^s$ and then note that $\psi(x)=\epsilon(\alpha^s x)$, so that using $\psi$ instead of $\epsilon$ has the effect of shifting the sequence cyclically by $s$ places.  Thus we get no new sequences in this way.  And we do not use the trivial additive character, as it produces a sequence all of whose terms are $1$, which does not have good correlation properties.}  For the definition of a generic additive character sequence, see Section \ref{Albert}.

The full set of m-sequences of length $p^n-1$ is obtained as we let $\alpha$ run through the primitive elements of $\Fpn$ and let $s$ run through all possible cyclic shifts.
Since $\Tr(x^p)=\Tr(x)$, we see that $\epsilon(x^p)=\epsilon(x)$, so the Galois sequence associated to $\alpha^p$ is the same as the Galois sequence associated to $\alpha$.
We say that two elements $a, b \in \Fpn$ are {\it Galois conjugates} to mean that $a=b^{p^k}$ for some $k \in \Z$.
If one wants to list all m-sequences of length $p^n$ without repetition, one should only let the primitive element $\alpha$ run through a set of representatives of the primitive elements of $\Fpn$ modulo the Galois conjugate relation, and should only let the shift $s$ run through $0,1,\ldots,p^n-2$.
The field $\Fpn$ has $\phi(p^n-1)/n$ distinct representatives of primitive elements modulo the Galois conjugate relation, where $\phi$ is the Euler phi-function, so there are are that many distinct Galois sequences of length $p^n-1$ and $p^n-1$ times as many m-sequences.

If we have two Galois sequences of length $p^n-1$, say $f$ associated to primitive element $\alpha$ and $g$ associated to primitive element $\beta$, then there is some integer $d$ with $\gcd(d,p^n-1)=1$ such that $\beta=\alpha^d$.  Then we say that {\it $g$ is the decimation by $d$ of $f$} or that {\it the pair $(f,g)$ has relative decimation $d$}, because the $k$th term of $g$ is $\epsilon(\alpha^{d k})$ which is the $(d k)$th term of $f$ (where we read term indices modulo $p^n-1$), and so $g$ can be obtained by taking every $d$th term of $f$ (proceeding cyclically modulo $p^n-1$).
Since $\epsilon(x^p)=\epsilon(x)$ as noted above, decimation by a power of $p$ does not change an m-sequence, and so we say that a decimation $d$ is {\it trivial} if $d \equiv p^k \pmod{p^n-1}$ for some $k \in \Z$.
If $d \equiv -p^k \pmod{p^n-1}$ for some $k \in \Z$, we say that $d$ is a {\it reversing decimation}, for then decimation by $d$ has the same effect as decimation by $-1$, which is reversal of the sequence (considered cyclically).

A generic m-sequence is a cyclic shift of a Galois sequence.
If $f$ and $g$ are generic m-sequences of the same length $p^n-1$, we say the pair {\it $(f,g)$ has relative decimation $d$} to mean that if $f$ is associated to primitive element $\alpha$ and $g$ is associated to primitive element $\beta$, then $\beta$ is a Galois conjugate of $\alpha^d$.
Equivalently, $(f,g)$ has relative decimation $d$ if we can obtain $g$ by taking every $d$th element of $f$ (proceeding cyclically modulo $p^n-1$) and then applying some cyclic shift.
Thus when we say that $(f,g)$ has relative decimation $d$, we are not specifying anything about the cyclic shifts of $f$ or $g$.
Just as with Galois sequences, a {\it trivial} relative decimation is a power of $p$ modulo $p^n-1$, and a {\it reversing} relative decimation is the negative of a power of $p$ modulo $p^n-1$.

\subsection{Known Autocorrelation Properties}\label{Joseph}

It was shown in \cite[Section 5]{Jensen-Hoholdt} and \cite[Theorem 2.2]{Jensen-Jensen-Hoholdt} (for binary sequences) and \cite[Theorem 1.3]{Katz-2013} (in general) that the autocorrelation merit factor of m-sequences tends to $3$ as their length tends to infinity, so we say that they have asymptotic autocorrelation merit factor $3$, or equivalently, asymptotic autocorrelation demerit factor $1/3$.
This is considerably better autocorrelation performance than one gets for a randomly selected sequence, since \eqref{David} shows that a random binary sequence of length $\ell$ has average autocorrelation demerit factor $1-1/\ell$.

\subsection{Known Crosscorrelation Properties}

For typical applications, we use m-sequence pairs $(f,g)$ where $g$ is not a cyclic shift of $f$, which is equivalent to saying that $f$ and $g$ are associated to primitive elements that are not Galois conjugates of each other, and this in turn is equivalent to saying that $(f,g)$ has nontrivial relative decimation.
If $(f,g)$ is a pair of m-sequences of length $\ell$ with nontrivial decimation, and if we use $T^k f$ to denote the sequence obtained from $f$ by cyclically shifting it $k$ places to the left, then Sarwate \cite[eq.~(54)]{Sarwate} shows that the average crosscorrelation demerit factor is
\begin{equation}\label{Sophia}
\frac{1}{\ell^2} \sum_{j,k \in \{0,1,\ldots,\ell-1\}} \CDF(T^j f, T^k g) = 1 + \frac{2}{3\ell} - \frac{1}{\ell^2} + \frac{1}{3 \ell^3},
\end{equation}
which is about $1$ for large $\ell$, and so about the same as the average over all binary sequences (see \eqref{Earl}).

We are thus tempted to conclude that although m-sequences have autocorrelation behavior superior to that of randomly selected sequences, nonetheless pairs of m-sequences typically have crosscorrelation behavior on par with pairs of randomly selected sequences.
This paper shows that in many cases, this judgement is true, but there are important exceptions where m-sequences can significantly exceed the performance of randomly selected sequences both in autocorrelation and crosscorrelation.

\subsection{Distribution of Crosscorrelation Demerit Factors}

It becomes apparent that something interesting is happening with the crosscorrelation of m-sequences when one looks at the distribution of crosscorrelation demerit factors for m-sequences of a given length.  We computed the crosscorrelation demerit factors for all pairs of m-sequences of length $255$ with nontrivial relative decimation, and we display the results in Figure \ref{Andrew}.
\begin{center}
\begin{figure}
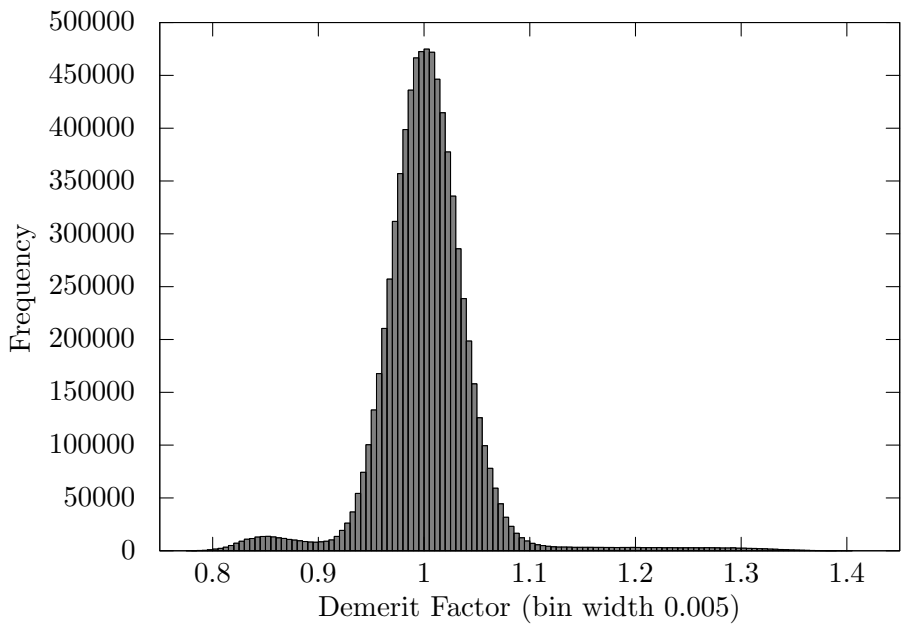

\begin{center}
\caption{Distribution of crosscorrelation demerit factors for all pairs of m-sequences of length $255$ with nontrivial relative decimation}\label{Andrew}

\end{center}
\end{figure}
\end{center}

One can see that although most of the demerit factors are concentrated in a peak centered at about $1$, there is a satellite peak centered at about $0.85$.
One finds that this satellite peak is due to sequence pairs with reversing relative decimation: when we separate the distribution of Figure \ref{Andrew} into those contributions from sequence pairs with reversing relative decimation and those with nonreversing relative decimation (see Figure \ref{Bartholomew}), we see that both the satellite peak and also the tail to the right of the main peak in Figure \ref{Andrew} are accounted for by reversing decimations.
So reversing relative decimation is responsible for crosscorrelation demerit factors both markedly superior and markedly inferior to the typical case.
This phenomenon is not detectable if one averages the crosscorrelation demerit factor over cyclic shifts, as Sarwate does (recall equation \eqref{Sophia} above).
\begin{center}
\begin{figure}
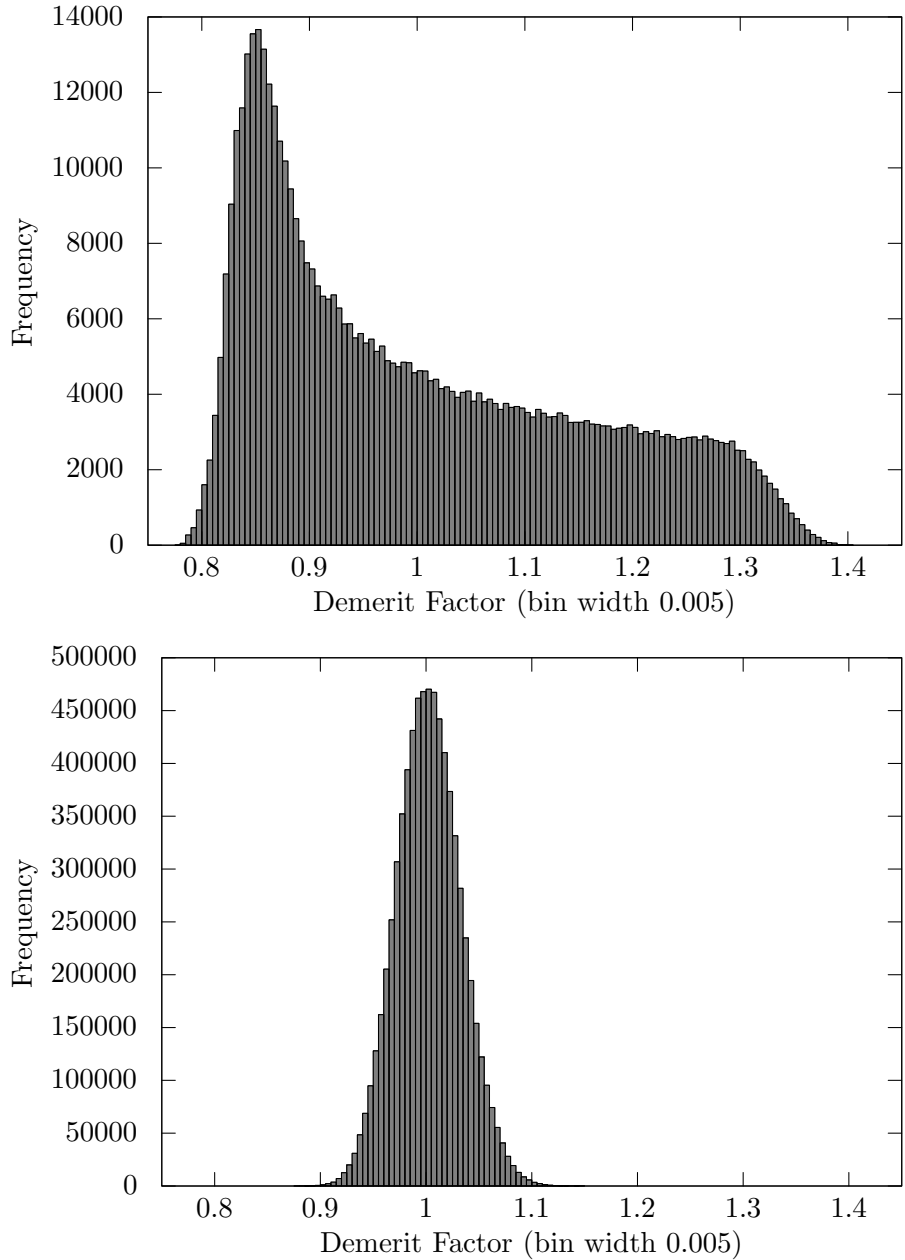

\begin{center}
\caption{Partitioned distribution of crosscorrelation demerit factors for pairs of m-sequences of length $255$ with nontrivial relative decimation.  Pairs with reversing relative decimation are shown on the top figure, while pairs with nonreversing relative decimation are shown on the bottom figure.  (Vertical scales differ to reveal the shape of the former.)}\label{Bartholomew}

\end{center}
\end{figure}
\end{center}

We want to determine the behavior of crosscorrelation demerit factor without averaging over cyclic shifts, especially to explain the reason why the reversing decimation produces a wider range of behavior than other nontrivial decimations.
Lack of averaging makes the calculations become much more difficult, and we only determine the asymptotic behavior.

\subsection{New Results: Asymptotic Crosscorrelation}

In this paper we show that for a fixed nontrivial, non-reversing decimation, the crosscorrelation merit factor tends to $1$ as the length of our m-sequences tends to infinity.
\begin{theorem}\label{Thomas}
Let $p$ be a prime and let $d$ be an integer such that $|d|$ is not a power of $p$.
Let $\{(f_\iota,g_\iota)\}_{\iota \in I}$ be an infinite family of pairs of $p$-ary m-sequences, where for each $\iota \in I$, the sequences $f_\iota$ and $g_\iota$ have the same length $q_\iota-1$ with $\gcd(d,q_\iota-1)=1$ and the pair $(f_\iota,g_\iota)$ has relative decimation $d$.
In the limit as $q_\iota \to \infty$, we have $\CDF(f_\iota,g_\iota) \to 1$.
\end{theorem}
On the other hand, for trivial and reversing decimations, the asymptotic crosscorrelation merit factors depend on the relative shifts of the sequences.  Full results for both these kinds of decimations are found in our main result, Theorem \ref{Victor} below, which also includes Theorem \ref{Thomas} above as a special case.  Here we present the result for the reversing decimation, since the trivial decimation is not suitable for most applications.
\begin{theorem}\label{Deidre}
Let $p$ be a prime.
Let $\{(f_\iota,g_\iota)\}_{\iota \in I}$ be an infinite family of pairs of $p$-ary m-sequences, where for each $\iota \in I$, the sequences $f_\iota$ and $g_\iota$ have the same length $q_\iota-1$.
Suppose that each $f_\iota$ is the m-sequence associated to some $\alpha_\iota$ with shift $r_\iota$ and $g_\iota$ is the m-sequence associated to $\alpha_\iota^{-1}$ with shift $s_\iota$, so that $(f_\iota,g_\iota)$ has reversing relative decimation.
If there is a real number $\Sigma$ such that $(r_\iota+s_\iota)/(q_\iota-1) \to \Sigma$ as $q_\iota \to \infty$, then in this limit we also have
\[
\CDF(f_\iota,g_\iota) \to \frac{1}{3}+(\Sigma^\prime-\floor{\Sigma^\prime})^2+(\Sigma^\prime-\floor{\Sigma^\prime}-1)^2,
\]
where $\floor{\cdot}$ is the floor function and where $\Sigma^\prime=\Sigma$ if $p=2$ and $\Sigma^\prime=\Sigma+1/2$ if $p$ is odd.
Thus the asymptotic crosscorrelation demerit factor varies from $5/6$ to $4/3$, with the minimum value of $5/6$ achieved whenever $\Sigma^\prime$ is half an odd integer.
\end{theorem}
This theorem shows that we can achieve an asymptotic crosscorrelation merit factor of $6/5$ with m-sequence pairs where one sequence in each pair is a suitably shifted reverse of the other.
This gives families of sequence pairs where both the autocorrelation of the individual sequences and the crosscorrelation of the pairs are asymptotically better than randomly sequences.
Since the asymptotic autocorrelation merit factor of m-sequences is $3$ (see Section \ref{Joseph}), Theorem \ref{Deidre} gives us families $\{(f_\iota,g_\iota)\}_{\iota \in I}$ of sequence pairs where
\[
\CDF(f_\iota,g_\iota)+\sqrt{\DF(f_\iota) \DF(g_\iota)} \to \frac{7}{6}
\]
in the limit as the length of the sequences tends to infinity.
This comes quite close to the lower bound of Pursley-Sarwate in \eqref{Gary}.

Recall Sarwate's result (see \eqref{Sophia} above) showing that if we average over all cyclic shifts, the mean crosscorrelation demerit factor is close to $1$ (for long sequences).
Theorem \ref{Thomas} suggests that for nontrivial, nonreversing decimations, the demerit factors are likely to be close to $1$ for all shifts, while Theorem \ref{Deidre} suggests that for reversing decimations, this average comes about from a balance of sequences with demerit factors significantly lower than $1$ and sequences with demerit factors significantly higher than $1$.  (Indeed the average value of the formula for $\CDF$ in Theorem \ref{Deidre} is $1$ as $\Sigma^\prime$ runs through $[0,1]$.)

We illustrate the dependence on shift in the crosscorrelation demerit factor of pairs of binary m-sequences in Figure \ref{Lester}.
In order to discuss this figure, we call the quantity $(r_\iota+s_\iota)/(q_\iota-1)$ from Theorem \ref{Deidre} the {\it fractional sum of shifts}, and its limit $\Sigma$ is the {\it limiting fractional sum of shifts}.
Note that Theorem \ref{Deidre} is only sensitive to the value of $\Sigma$ modulo $1$, and so we often report it modulo $1$.

In the top graph in Figure \ref{Lester}, the data points correspond to crosscorrelation demerit factors for $511$ pairs of binary m-sequences of length $511$, where in each pair the first sequence is the Galois sequence associated to a root of $x^9+x^4+1$ over $\F_2$ and the second sequence ranges over the m-sequences obtained by decimating the first sequence by $3$ (a nontrivial, nonreversing decimation) and applying all possible shifts.  The crosscorrelation demerit factor is then shown as a function of the fractional sum of shifts.  The line in the top graph of Figure \ref{Lester} indicates the value of the asymptotic crosscorrelation demerit factor from Theorem \ref{Thomas}: a value of $1$, independent of shifting.  We see that the data agree well with the asymptotic values.

In the bottom graph of Figure \ref{Lester}, the data points show the crosscorrelation demerit factors as a function of fractional sum of shifts for a similar collection of m-sequence pairs, with the only difference being that the second sequence in each pair is obtained using the reversing decimation of $-1$ rather than the decimation $3$.   The curve in the bottom graph of Figure \ref{Lester} indicates the value of the asymptotic crosscorrelation demerit factor from Theorem \ref{Deidre}, this time showing dependence on shifting.  Again, the data agree well with the asymptotic values.

We note that the skewed distribution of crosscorrelation demerit factors in Figure \ref{Bartholomew} is an expected consequence of the parabolic shape of the curve in the bottom graph of Figure \ref{Lester}: there is a greater density of low demerit factors since the slope of the curve in Figure \ref{Lester} is gentle when the demerit factor is low and steep when the demerit factor is high.
\begin{center}
\begin{figure}
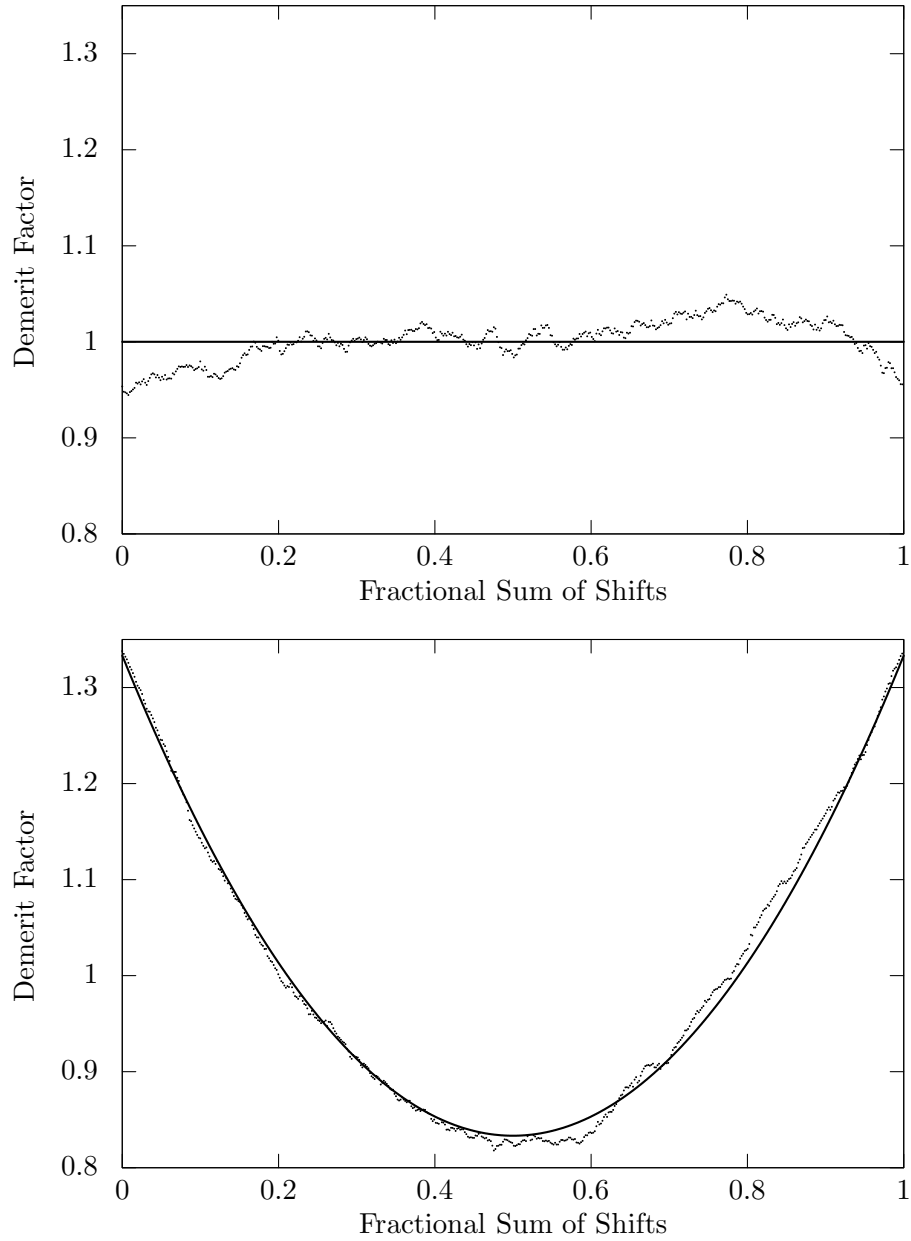

\begin{center}
\caption{Crosscorrelation demerit factor for binary m-sequence pairs as a function of fractional sum of shifts with nontrivial, nonreversing relative decimation (top) and reversing relative decimation (bottom).  Lines and curves are asymptotic values, points are data.}\label{Lester}

\end{center}
\end{figure}
\end{center}
\FloatBarrier

Figure \ref{Laura} repeats the illustration for ternary m-sequences, now of length $728$, with the first sequence in each pair always being the Galois sequence associated to a root of $x^6+x-1$ over $\F_3$ and the second sequence running through all shifts of a particular decimation of the first, with the nontrivial, nonreversing decimation of $5$ in the top graph and the reversing decimation $-1$ in the bottom graph.
Figure \ref{Laura} illustrates the same behavior as Figure \ref{Lester}, except that (in the case of reversing relative decimation) the dependence of crosscorrelation demerit factor on fractional shift is itself translated by $1/2$ (modulo $1$), just as predicted for odd characteristic by Theorem \ref{Deidre}.
This shift is clearly understood in the derivation of Theorem \ref{Deidre} as a consequence of a fundamental character sum that underlies the calculation (worked out in Lemma \ref{Edward}).
\begin{center}
\begin{figure}
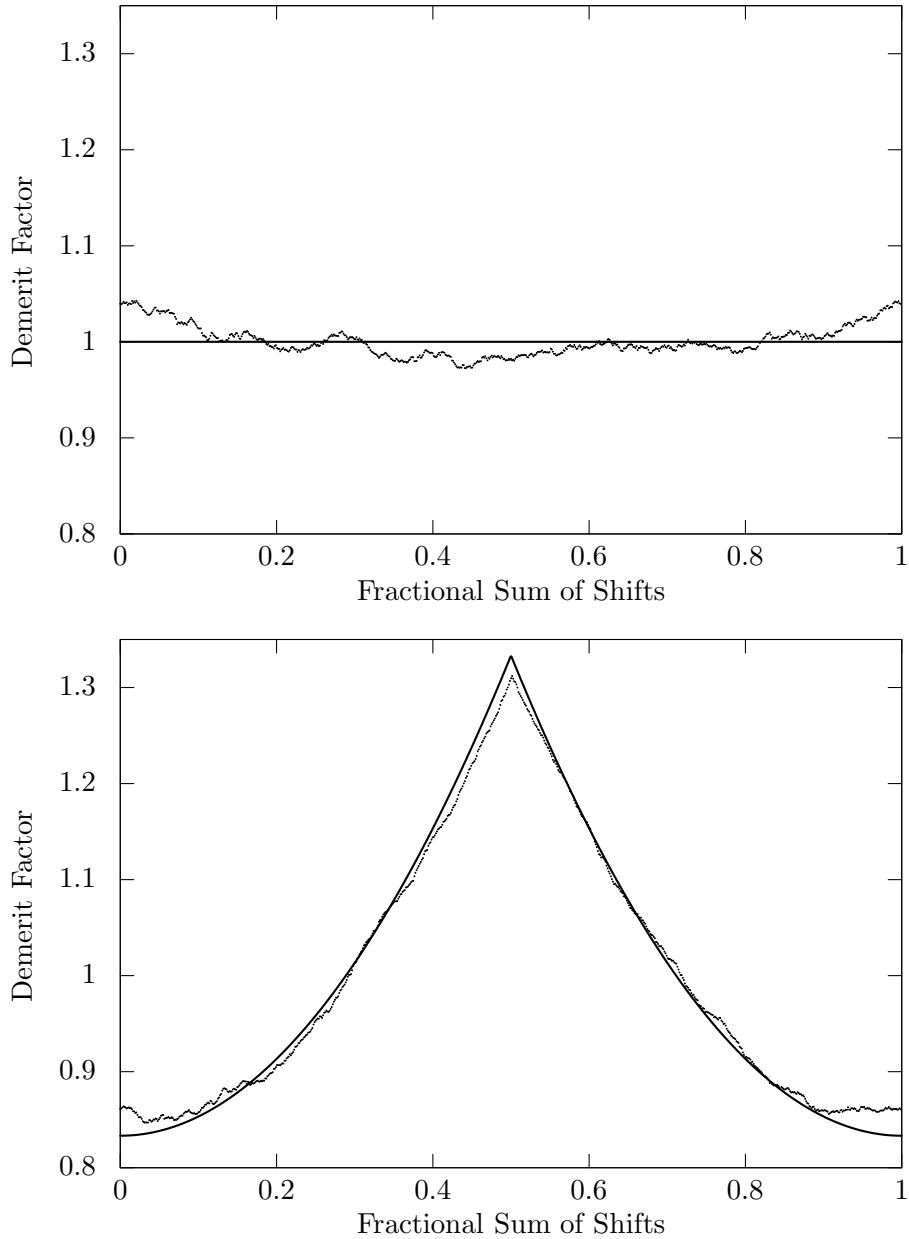

\begin{center}
\caption{Crosscorrelation demerit factor for ternary m-sequence pairs as a function of fractional sum of shifts with nontrivial, nonreversing relative decimation (top) and reversing relative decimation (bottom).  Lines and curves are asymptotic values, points are data.}\label{Laura}

\end{center}
\end{figure}
\end{center}
\FloatBarrier

\section{Additive Character Sequences}\label{Gilbert}

It was shown in \cite[Theorem 1.3]{Katz-2013} and \cite[Theorem 2.2]{Jedwab-Katz-Schmidt-2013-Advances} that one can modify m-sequences by a process known as appending (described in Section \ref{Albert} below) to achieve an autocorrelation merit factor as large as $3.342065\ldots$, the largest root of $7x^3-33 x^2+33 -3$.
These modified m-sequences are known as {\it additive character sequences}, and we now define them and then look at their crosscorrelation properties.

\subsection{Definition of an Additive Character Sequence}\label{Albert}

Additive character sequences are derived from m-sequences by modifying their length.
On the one hand, one can {\it truncate} an m-sequence of length $q-1=p^n-1$ by only retaining the first $\ell$ terms where $\ell < q-1$.
Or one can {\it append} the same m-sequence by periodic extension, that is, we increase the number of terms, and each new term (say term $k$) is the same as term $k-(q-1)$.
A true m-sequence where one retains exactly $q-1$ terms is said to be an {\it additive character sequence of natural length}.
More generally, if $\alpha$ is a primitive element of some finite field $F$ of characteristic $p$ and order $q$, and if $s, \ell \in \Z$ with $\ell>0$, then the {\it $p$-ary additive character sequence associated to $\alpha$ with shift $s$ and length $\ell$} is
\[
(\epsilon(\alpha^{s}),\ldots,\epsilon(\alpha^{s+\ell-1})),
\]
where $\epsilon$ is the canonical additive character of $F$.
So the $p$-ary $m$-sequence associated to $\alpha$ with shift $s$ is just the additive character sequence associated to $\alpha$ with shift $s$ and length $\card{\Fu}=q-1$.
The {\it fractional length} of an additive character sequence is the ratio of its length to the length of an m-sequence associated to the same element.
So for our additive character sequence associated to the primitive element $\alpha$ of $F$ with shift $s$ and length $\ell$, the fractional length is $\ell/\card{\Fu}$, and an m-sequence is precisely an additive character sequence of fractional length $1$.

\subsection{Known Autocorrelation Properties}

Recall from Section \ref{Joseph} that it was shown in \cite[Section 5]{Jensen-Hoholdt}, \cite[Theorem 2.2]{Jensen-Jensen-Hoholdt}, and \cite[Theorem 1.3]{Katz-2013} that m-sequences have an asymptotic autocorrelation merit factor of $3$.
Appending can be used to raise the autocorrelation merit factor of m-sequences to $3.342065\ldots$, the largest root of $7 x^3-33 x^2+33 x-3$ (see \cite[Theorem 1.3]{Katz-2013}, \cite[Theorem 2.2]{Jedwab-Katz-Schmidt-2013-Advances}).

\subsection{New Results: Asymptotic Crosscorrelation}

Our theory enables us to compute the asymptotic crosscorrelation demerit factors for pairs of additive character sequences, just as we computed them for m-sequences in Theorems \ref{Thomas} and \ref{Deidre} above.
The asymptotic crosscorrelation demerit factor depends both on fractional length and shifting.
To express our result, we introduce the function
\begin{equation}\label{Orlando}
\Omega(x,y)=\sum_{n \in \Z} \max(0,1-|n x- y|)^2,
\end{equation}
which we define on $\{(x,y) \in \R^2: x\not=0\}$.
Although $\Omega(x,y)$ appears to involve an infinite sum, only finitely many terms are nonzero for any given point $(x,y)$ in its domain, and $\Omega$ is continuous on its whole domain.

The following is our main result on asymptotic crosscorrelation for additive character sequences, proved in Section \ref{Simon}.
\begin{theorem}\label{Victor}
Let $p$ be a prime and let $d$ be an integer.
Let $\{(f_\iota,g_\iota)\}_{\iota \in I}$ be an infinite family of pairs of sequences, where for each $\iota \in I$, the sequence $f_\iota$ is the $p$-ary additive character sequence associated to primitive element $\alpha_\iota$ of field $F_\iota$ with shift $r_\iota$ and length $\ell_\iota$, while $g_\iota$ is the $p$-ary additive character sequence associated to primitive element $\alpha_\iota^d$ of field $\Fi$ with shift $s_\iota$ and length $\ell_\iota$.
Let $q_\iota=\card{\Fi}$ for each $\iota \in I$, and note that the primitivity of $\alpha_\iota^d$ means that $\gcd(d,q_\iota-1)=1$.
Suppose $\{q_\iota\}_{\iota \in I}$ is unbounded and that there is a positive real number $\Lambda$ such that in the limit as $q_\iota \to \infty$, we have $\ell_\iota/(q_\iota-1) \to \Lambda$.
\begin{enumerate}[(i).]
\item\label{Alan} If $d$ is a power of $p$ (including $d=p^0=1$), and if $(r_\iota-s_\iota)/(q_\iota-1)$ tends to a real number $\Delta$ as $q_\iota \to \infty$, then
\[
\CDF(f_\iota,g_\iota) \to -\frac{2}{3} \Lambda + \Omega\left(\frac{1}{\Lambda},0\right) + \Omega\left(\frac{1}{\Lambda},\frac{\Delta}{\Lambda}\right)
\]
as $q_\iota \to \infty$.
\item\label{Reginald} If $-d$ is a power of $p$ (including $d=-p^0=-1$), and if $(r_\iota+s_\iota)/(q_\iota-1)$ tends to a real number $\Sigma$ as $q_\iota \to \infty$, then
\[
\CDF(f_\iota,g_\iota) \to -\frac{2}{3} \Lambda + \Omega\left(\frac{1}{\Lambda},0\right) + \Omega\left(\frac{1}{\Lambda},1+\frac{\Sigma^\prime}{\Lambda}\right)
\]
as $q_\iota \to \infty$, where $\Sigma'=\Sigma$ if $p=2$ and $\Sigma'=\Sigma+1/2$ if $p$ is odd.
\item\label{Ophelia} If $|d|$ is not a power of $p$, then
\[
\CDF(f_\iota,g_\iota) \to \Omega\left(\frac{1}{\Lambda},0\right)
\]
as $q_\iota \to \infty$.
\end{enumerate}
\end{theorem}
The global minima of these three limiting formulae for asymptotic crosscorrelation demerit factor are found in Lemmata \ref{Jake}--\ref{Leonard}, which give the values of parameters $\Lambda$, $\Delta$, and $\Sigma^\prime$  where these minima occur.
We then look at the reciprocals of these minima to get the global maximum asymptotic crosscorrelation merit factor:
\begin{itemize}
\item In parts \eqref{Alan} and \eqref{Reginald} we achieve a maximum asymptotic crosscorrelation merit factor of $1.539389\ldots$, the largest root of the polynomial $19 x^3-54 x^2+42 x-6$.
\item In part \eqref{Ophelia} we achieve a maximum asymptotic crosscorrelation merit factor of $1$.
\end{itemize}

Note that Theorems \ref{Thomas} and \ref{Deidre} are just specializations of parts \eqref{Ophelia} and \eqref{Reginald}, respectively, of Theorem \ref{Victor} since m-sequences are additive character sequences of natural length (so $\Lambda=1$).  Furthermore, the known results for autocorrelation of additive character sequences (\cite[Theorem 1.3]{Katz-2013}, \cite[Theorem 2.2]{Jedwab-Katz-Schmidt-2013-Advances}) are specializations of part \eqref{Alan} where we set $f_\iota=g_\iota$ (so then $\Delta=0$).

\subsection{Appending to Minimize Autocorrelation}\label{Francis}

One can increase the autocorrelation merit factor of m-sequences by appending: the asymptotic merit factor is maximized (see \cite[Theorem 1.3]{Katz-2013} \cite[Theorem 2.2]{Jedwab-Katz-Schmidt-2013-Advances})  when one uses families of m-sequences whose fractional length tends to $\Lambdaapp=1.115749\ldots$, the middle root of $x^3-12 x+12$.  In this case, one obtains an asymptotic autocorrelation merit factor equal to $3.342065\ldots$, the largest root of $7 x^3-33 x^2+33 x-3$.
Although this appending raises the autocorrelation merit factor above the asymptotic value of $3$ for m-sequences of natural length, we shall now see that this appending decreases crosscorrelation performance.

For sequence pairs with nontrivial, nonreversing relative decimations, Theorem \ref{Victor}\eqref{Ophelia} shows that appending to limiting fractional length of $\Lambdaapp=1.115749\ldots$ increases the limiting crosscorrelation demerit factor to a value of $1.021524\ldots$, thus reducing the corresponding crosscorrelation merit factor to $0.978928\ldots$, which is worse than the asymptotic value of $1$ that is obtained with natural length.
And indeed a slight worsening of average crosscorrelation performance can be seen by comparing the demerit factors for the m-sequences of natural length $511$ in the top graph of Figure \ref{Lester} with those obtained by appending them to length $570\approx 511 \cdot \Lambdaapp$ in the top graph of Figure \ref{Linus}.  Note that in both cases, no significant dependence on shift is seen, and indeed the asymptotic value has no shift-dependence, as seen in Theorem \ref{Victor}\eqref{Ophelia}.

Now consider sequence pairs with reversing relative decimation, where we append to limiting fractional length $\Lambdaapp=1.115749\ldots$, shift our sequences to minimize the asymptotic crosscorrelation demerit factor (letting $\Sigma'=-\Lambdaapp+1/2$ achieves this: see Lemma \ref{Rachel} for details).
In this case, Theorem \ref{Victor}\eqref{Reginald} shows that the limiting crosscorrelation demerit factor becomes $0.886814\ldots$, corresponding to asymptotic crosscorrelation merit factor of $1.127631\ldots$.
This is again worse than the asymptotic value of $6/5$ obtained for pairs of $m$-sequences of natural length with reversing decimations.
If we want to compare the performance of these sequence pairs with the lower bound of Pursley-Sarwate in \eqref{Gary}, we have
\[
\CDF(f_\iota,g_\iota)+\sqrt{\DF(f_\iota) \DF(g_\iota)} \to 1.1860303\ldots
\]
in the limit as the length of the sequences tends to infinity.

If one wants to see performance over a range of $\Sigma$ values, then one can compare the performance of m-sequences of natural length $511$ in the bottom graph of Figure \ref{Lester} with those obtained by appending them to length $570 \approx 511 \cdot \Lambdaapp$ in the bottom graph of Figure \ref{Linus}.  One sees that appending has caused a translation in the dependence on fractional shift and a decrease in the amplitude of the variation: the optimal performance (shifting to minimize demerit factor) is worse, while the worst-case performance is better.
\begin{center}
\begin{figure}
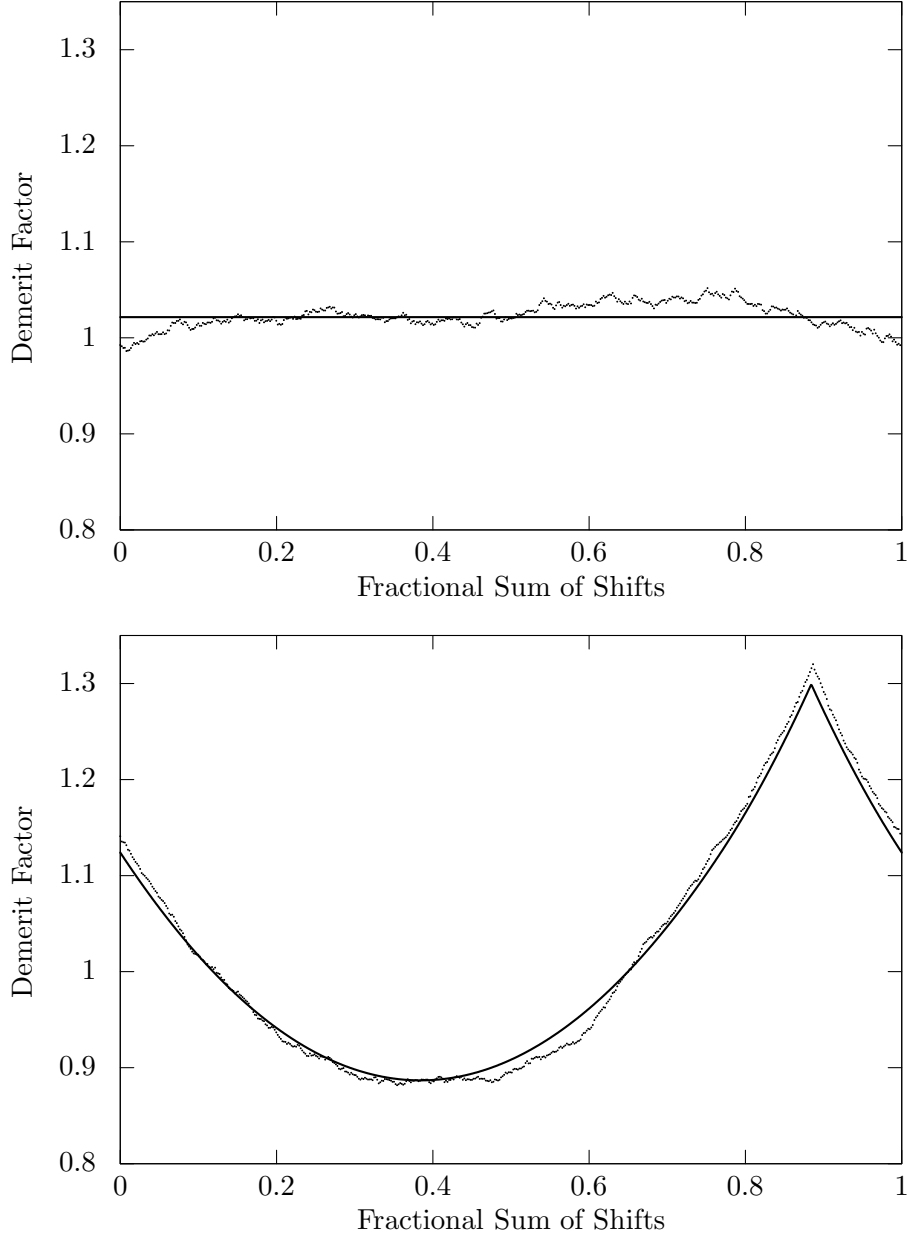

\begin{center}
\caption{Crosscorrelation demerit factor for appended binary m-sequence pairs as a function of fractional sum of shifts with nontrivial, nonreversing relative decimation (top) and reversing relative decimation (bottom).  Lines and curves are asymptotic values from Theorem \ref{Victor}, points are data.}\label{Linus}

\end{center}
\end{figure}
\end{center}
\FloatBarrier

\subsection{Truncating to Minimize Crosscorrelation}\label{Timothy}

Let us also discuss the effect of truncation.
For nontrivial, nonreversing relative decimations, Theorem \ref{Victor}\eqref{Ophelia} shows that the asymptotic crosscorrelation merit factor of truncated sequences is always $1$, which is precisely the same as for m-sequences of natural length.  For such a decimation, compare the performance of m-sequences of natural length $511$ in the top graph of Figure \ref{Lester} with those obtained by truncating them to length $285$ in the top graph of Figure \ref{Lisa}.  On the other hand, autocorrelation performance deteriorates with truncation (use \cite[Theorem 1.3]{Katz-2013} or Theorem \ref{Victor}\eqref{Alan} with $\Delta=0$ to see this), so there is no compelling reason to truncate m-sequence pairs with nontrivial, nonreversing decimations.

With reversing relative decimations, truncation can increase the crosscorrelation merit factor significantly beyond what is achievable with pairs of natural length m-sequences.
In Lemma \ref{Barbara} we show that the crosscorrelation merit factor, which is the reciprocal of the crosscorrelation demerit factor given in Theorem \ref{Victor}\eqref{Reginald}, achieves a global maximum value of $1.539389\ldots$, the largest root of the polynomial $19 x^3-54 x^2+42 x-6$, when we let the fractional length of our sequences tend to $\Lambdatrunc=0.557874\ldots$, the middle root of $2 x^3-6 x+3$, and shift our sequences so that the fractional sum of shifts tends to $1/2-\Lambdatrunc$ as the length of the sequences tends to infinity.\footnote{This is for binary sequences, where the factional sum of shifts $\Sigma$ equals $\Sigma'$ in Theorem \ref{Victor}.  For nonbinary sequences, we need $\Sigma' \to 1/2-\Lambdatrunc$ so that $\Sigma \to -\Lambdatrunc$ or $1-\Lambdatrunc$ (recall that we read these values modulo $1$).  Note also that $\Lambdatrunc=\Lambdaapp/2$ where $\Lambdaapp=1.115749\ldots$ is the limiting fractional length that maximizes asymptotic autocorrelation merit factor (as discussed in Section \ref{Francis}).}
To see this effect, compare the performance m-sequences of natural length $511$ in the bottom graph of Figure \ref{Lester} with those obtained by truncating them to length $285\approx 511\cdot \Lambdatrunc$ in the bottom graph of Figure \ref{Lisa}.
We should also note that this truncation does decrease the asymptotic autocorrelation merit factor of the sequences used to $1.592144\ldots$ (as compared to $3$ for m-sequences of natural length).  So one sees that there is a tradeoff between autocorrelation and crosscorrelation performance, and one can use truncation or appending to strike a balance between autocorrelation and crosscorrelation performance appropriate for one's specific application.
If we want to compare the performance of these sequence pairs with the lower bound of Pursley-Sarwate in \eqref{Gary}, we have
\[
\CDF(f_\iota,g_\iota)+\sqrt{\DF(f_\iota) \DF(g_\iota)} \to 1.277691\ldots
\]
in the limit as the length of the sequences tends to infinity.
\begin{center}
\begin{figure}
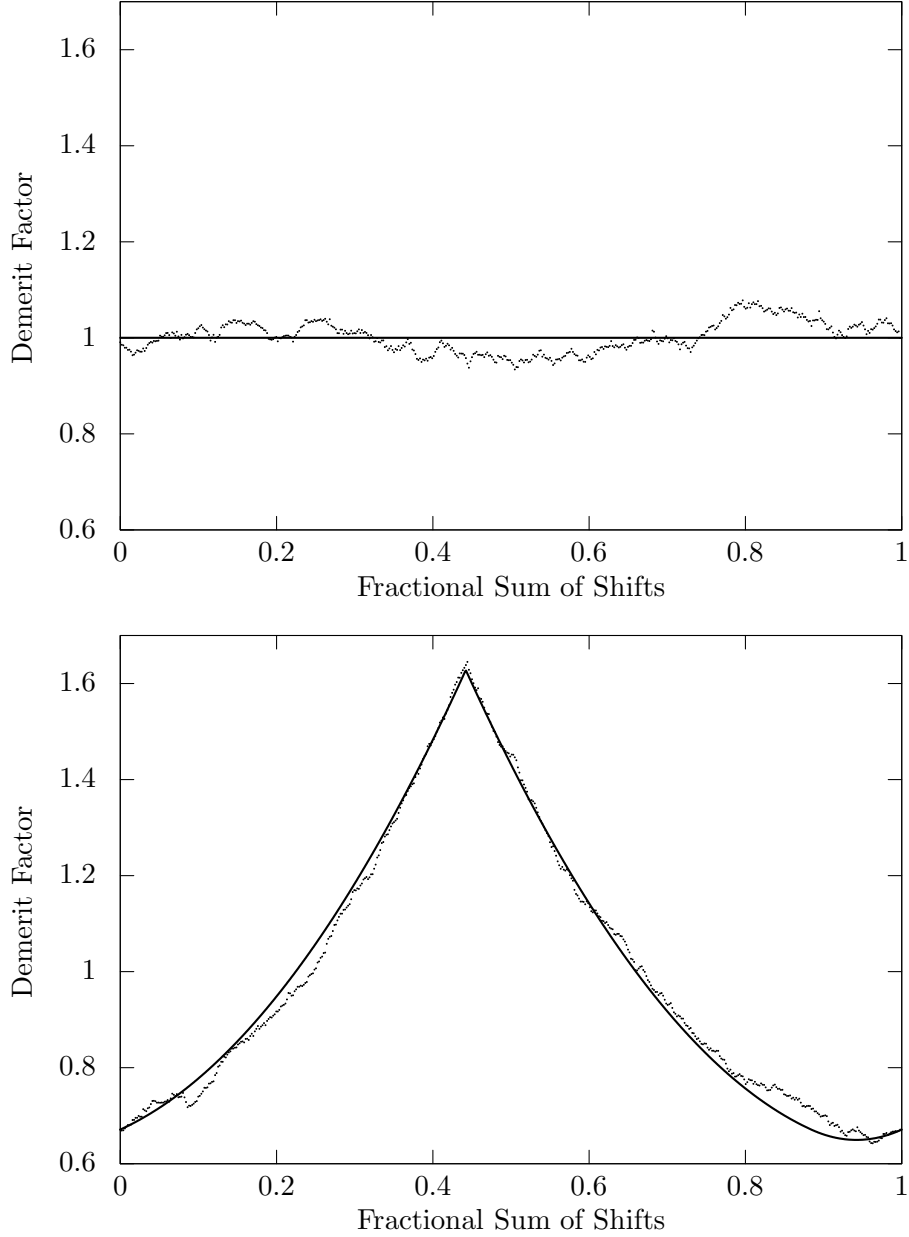

\begin{center}
\caption{Crosscorrelation demerit factor for truncated binary m-sequence pairs as a function of fractional sum of shifts with nontrivial, nonreversing relative decimation (top) and reversing relative decimation (bottom).  Lines and curves are asymptotic values from Theorem \ref{Victor}, points are data.}\label{Lisa}

\end{center}
\end{figure}
\end{center}
\FloatBarrier

\section{Multiplicative Character Sequences}\label{Milton}

\subsection{Definition of a Multiplicative Character Sequence}

Now we discuss a family of sequences that are in some sense dual to the additive character sequences.
We let $p$ be a prime, and let $\chi\colon \Fpu \to \C$ be a nontrivial multiplicative character.
If $\chi$ is of order $m$, then its values are complex $m$th roots of unity.
So binary $(\pm 1)$ sequences are produced when $\chi$ is of order $2$, that is, when $\chi$ is the quadratic character (Legendre symbol).
We extend the domain of $\chi$ to be all of $\Fp$ by setting $\chi(0)=0$, as is traditional for multiplicative characters.
If $s,\ell \in \Z$ with $\ell >0$, then the {\it multiplicative character sequence associated to character $\chi$ with shift $s$ and length $\ell$} is
\begin{equation}\label{Nicholas}
(\chi(s),\chi(s+1),\ldots,\chi(s+\ell-1)).
\end{equation}
We say that the sequence has {\it natural shift} if $s=0$.
If $\ell=p$, then we say that the sequence is of {\it natural length}.
If $\ell < p$, we say that it is {\it truncated}, and if $\ell > p$, we say that it is {\it appended}.
When $\chi$ is the quadratic character (Legendre symbol), we call our multiplicative character sequence a {\it quadratic character sequence}.
If $\chi$ is not the quadratic character, then we call our multiplicative character sequence a {\it nonquadratic character sequence}.
The {\it fractional length} of a multiplicative character sequence associated to character $\chi$ of length $\ell$ is $\ell/p$, where $p$ is the order of the field for which $\chi$ is a character.
So a multiplicative character sequence of natural length is precisely a multiplicative character sequence of fractional length $1$.

In practical applications involving multiplicative character sequences such as \eqref{Nicholas}, one generally replaces those terms $\chi(j)$ where $p\mid j$ (which makes $\chi(j)=0$) with $1$ or another complex root of unity to obtain a sequence where all terms have unit magnitude.
When we modify a sequence to replace all zero terms with complex numbers of unit magnitude, we say that the modified sequence is a {\it unimodularization} of the original sequence.
We shall show (in Lemma \ref{Oswald}) that unimodularization of the above multiplicative character sequences does not change their asymptotic autocorrelation and crosscorrelation merit factors, so that all our asymptotic results work just as well for the original multiplicative character sequences as for their unimodularizations (so we may omit the distinction in this context).

A {\it Legendre sequence} is a unimodularized quadratic character sequence of natural shift and natural length, and a {\it shifted Legendre sequence} is a unimodularized quadratic character sequence of natural length and arbitrary shift.
The excellent autocorrelation properties of shifted Legendre sequences were observed by Turyn and reported by Golay \cite{Golay-1983}, and their asymptotic merit factor was computed by H\o holdt and Jensen \cite{Hoholdt-Jensen}.

We said that our multiplicative character sequences are in some sense duals of additive character sequences.  To see this, recall that one obtains an additive character sequence by selecting a generator $\alpha$ of the multiplicative group $\Fpnu$ of a finite field $\Fpn$; then one lists the elements of $\Fpnu$ in {\it multiplicative order}, that is, as successive powers $\alpha^0,\alpha^1,\ldots,\alpha^{p^n-2}$ of the generator; and then one applies a nontrivial {\it additive character} $\psi$ to each of these powers, thus yielding the m-sequence $(\psi(\alpha^0),\ldots,\psi(\alpha^{p^n-2}))$, which can then be shifted and truncated or appended to get an additive character sequence.  One obtains a multiplicative character sequence by using the generator $1$ of the additive group\footnote{One could use any other nonzero element $c \in \Fp$ as a generator instead of $1$.  This change would scale every term of the sequence by $\chi(c)$, a unimodular complex number, and this would have no effect on the magnitude of autocorrelation or crosscorrelation values.} of the finite field $\Fp$; then listing the elements of $\Fp$ in {\it additive order}, that is, as successive multiples $0,1,\ldots,p-1$ of the generator; and then one applies a nontrivial {\it multiplicative character} $\chi$ to each of these powers, thus yielding the multiplicative character sequence $(\chi(0),\ldots,\chi(p-1))$ of natural shift and length, which can then be shifted and truncated or appended to get a multiplicative character sequence.  The interplay between additive and multiplicative structure of finite fields gives both sequence constructions their desirable properties: their pseudorandomness arises from the fact that the type of character used is opposite to the type of generator used.

\subsection{Known Autocorrelation Properties}

In \cite{Hoholdt-Jensen} H\o holdt and Jensen consider Legendre sequences and shifted Legendre sequences (that is, unimodularized quadratic character sequences of natural length).  They show that the asymptotic autocorrelation merit factor is dependent on how the sequences are shifted, and ranges from $3/2$ (for non-shifted sequences, known as Legendre sequences) to $6$ (for appropriately shifted versions).
Later Borwein and Choi \cite[Theorem 3.2]{Borwein-Choi-2000} showed that the asymptotic autocorrelation merit factor of nonquadratic multiplicative character sequences of natural length is $3$, regardless of how they are shifted.
Later it was shown that appending can be used to raise the autocorrelation merit factor of unimodularized quadratic character sequences to $6.342061\ldots$, the largest root of $29 x^3-249 x^2+417 x -27$ (see \cite[Theorem 1.1]{Jedwab-Katz-Schmidt-2013-Littlewood}, \cite[Theorem 1.5]{Katz-2013}, and \cite[Theorem 2.1]{Jedwab-Katz-Schmidt-2013-Advances}) and can be used to raise the autocorrelation merit factor of nonquadratic multiplicative character sequences to $3.342065\ldots$, the largest root of $7 x^3-33 x+33 x-3$ (see \cite[Theorem 1.4]{Katz-2013}).

\subsection{New Results: Asymptotic Crosscorrelation}

Now we are ready to state our results for crosscorrelation merit factors of pairs of multiplicative character sequences.
Note that if $\chi \colon \Fpu \to \C$ is a multiplicative character, then the conjugate character $\conj{\chi}$ indicates the multiplicative character with $\conj{\chi}(a)=\conj{\chi(a)}=\chi(a)^{-1}$ for all $a \in \Fpu$.
Also recall the definition of the function $\Omega$ from \eqref{Orlando}.
As with autocorrelation, there is a marked difference in behavior between the quadratic character (Legendre symbol) and all other nontrivial multiplicative characters.

The following is our main result on asymptotic crosscorrelation for multiplicative character sequences, proved in Section \ref{Veronica}.
\begin{theorem}\label{Anthony}
Let $\{(f_\iota,g_\iota)\}_{\iota \in I}$ be an infinite family of pairs of sequences, where for each $\iota \in I$, the sequence $f_\iota$ is the multiplicative character sequence associated with character $\phi_\iota$ of prime field $F_\iota$ with shift $r_\iota$ and length $\ell_\iota$, and the sequence $g_\iota$ is the multiplicative character sequence associated with the character $\chi_\iota$ of field $F_\iota$ with shift $s_\iota$ and length $\ell_\iota$.
Let $p_\iota=\card{\Fi}$ for each $\iota \in I$, and suppose that $\{p_\iota\}_{\iota \in I}$ is unbounded and there is a positive real number $\Lambda$ such that in the limit as $p_\iota \to \infty$, we have $\ell_\iota/p_\iota \to \Lambda$.
\begin{enumerate}[(i).]
\item\label{Bernice} If $\phi_\iota=\chi_\iota$ is a nonquadratic character for each $\iota \in I$, and if $(r_\iota-s_\iota)/p_\iota$ tends to a real number $\Delta$ as $p_\iota \to \infty$, then
\[
\CDF(f_\iota,g_\iota) \to -\frac{2}{3} \Lambda + \Omega\left(\frac{1}{\Lambda},0\right) + \Omega\left(\frac{1}{\Lambda},\frac{\Delta}{\Lambda}\right)
\]
as $p_\iota \to \infty$.
\item\label{Nora} If $\phi_\iota=\conj{\chi_\iota}$ with $\phi_\iota$ a nonquadratic character for each $\iota \in I$, and if $(r_\iota+s_\iota)/p_\iota$ tends to a real number $\Sigma$ as $p_\iota \to \infty$, then
\[
\CDF(f_\iota,g_\iota) \to -\frac{2}{3} \Lambda + \Omega\left(\frac{1}{\Lambda},0\right) + \Omega\left(\frac{1}{\Lambda},1+\frac{\Sigma}{\Lambda}\right)
\]
as $p_\iota \to \infty$.
\item\label{Walter} If $\phi_\iota\not\in\{\chi_\iota, \conj{\chi_\iota}\}$ for each $\iota \in I$, then
\[
\CDF(f_\iota,g_\iota) \to \Omega\left(\frac{1}{\Lambda},0\right)
\]
as $p_\iota \to \infty$.
\item\label{Hubert} If $\phi_\iota=\chi_\iota$ is a quadratic character for each $\iota \in I$, if $(r_\iota+s_\iota)/p_\iota$ tends to a real number $\Sigma$ as $p_\iota \to \infty$, and if $(r_\iota-s_\iota)/p_\iota$ tends to a real number $\Delta$ as $p_\iota \to \infty$, then
\[
\CDF(f_\iota,g_\iota) \to -\frac{4}{3} \Lambda + \Omega\left(\frac{1}{\Lambda},0\right) + \Omega\left(\frac{1}{\Lambda},\frac{\Delta}{\Lambda}\right) + \Omega\left(\frac{1}{\Lambda},1+\frac{\Sigma}{\Lambda}\right)
\]
as $p_\iota \to \infty$.
\end{enumerate}
All the conclusions of this theorem remain the same if we replace each multiplicative character sequence with a unimodularization of that sequence.
\end{theorem}
The global minima of these four limiting formulae for asymptotic crosscorrelation demerit factor are found in Lemmata \ref{Jake}--\ref{Christopher}, which give the values of parameters $\Lambda$, $\Delta$, and $\Sigma^\prime$  where these minima occur.
We then look at the reciprocals of these minima to get the global maximum asymptotic crosscorrelation merit factor:
\begin{itemize}
\item In parts \eqref{Bernice} and \eqref{Nora} we achieve a maximum asymptotic crosscorrelation merit factor of $1.539389\ldots$, the largest root of the polynomial $19 x^3-54 x^2+42 x-6$.
\item In part \eqref{Walter} we achieve a maximum asymptotic crosscorrelation merit factor of $1$.
\item In part \eqref{Hubert}, we achieve a maximum asymptotic crosscorrelation merit factor of $3.342065\ldots$, the largest root of the polynomial $7 x^3-33 x^2+33 x-3$.
\end{itemize}

If we look at Theorem \ref{Anthony}, we see that part \eqref{Hubert} for quadratic characters is the analogue of both parts \eqref{Bernice} and \eqref{Nora} for nonquadratic characters since there is only one quadratic character $\eta$ in each prime field of odd order (and none in $\F_2$) and it has $\conj{\eta}=\eta$.
Also note that part \eqref{Walter} covers both quadratic and nonquadratic characters.
The previously known results for autocorrelation of multiplicative character sequences (see \cite[Theorems 1.4, 1.5]{Katz-2013}, \cite[Theorem 2.1]{Jedwab-Katz-Schmidt-2013-Littlewood}, and \cite[Theorem 2.1]{Jedwab-Katz-Schmidt-2013-Advances}) are specializations of Theorem \ref{Anthony}, specifically parts \eqref{Bernice} and \eqref{Hubert}, where we set $f_\iota=g_\iota$ (so then $\Delta=0$).

Comparison of parts \eqref{Bernice}, \eqref{Nora}, and \eqref{Walter} of Theorem \ref{Anthony} to the corresponding parts of Theorem \ref{Victor} show that sequences associated to nonquadratic multiplicative characters behave very much like binary additive character sequences: \eqref{Bernice} pairs of multiplicative character sequences associated to the same character behave like additive character sequence pairs with trivial relative decimation, \eqref{Nora} pairs of multiplicative character sequences associated to conjugate characters behave like additive character sequence pairs with reversing relative decimation, and \eqref{Walter} pairs of multiplicative character sequences associated to distinct, nonconjugate characters behave like additive character sequence pairs with nontrivial, nonreversing relative decimation.

\subsection{Quadratic Character Sequences}

As the asymptotic behavior of multiplicative character sequences associated to nonquadratic characters is so similar to that of additive character sequences, we focus on the quadratic character sequences, which are also especially interesting as they are binary ($\pm 1$) sequences, while all the other nontrivial multiplicative characters produce higher order roots of unity.
Recall that a {\it shifted Legendre sequence} is a unimodularized quadratic character sequence of natural length, and that a {\it Legendre sequence} is a shifted Legendre sequence with natural shift.
Since the zero term is typically replaced with a $1$ in the unimodularization, this produces a true binary sequence (all terms $\pm 1$).

We find that we can obtain a higher asymptotic crosscorrelation merit factor with truncated quadratic character sequences than we can with m-sequences.
To describe how to do this, we need some terminology: if $(f_\iota,g_\iota)$ is the sequence pair described in Theorem \ref{Anthony}, then $(r_\iota+s_\iota)/p_\iota$ is called the {\it fractional sum of shifts} and $(r_\iota-s_\iota)/p_\iota$ is called the {\it fractional difference of shifts}, and the parameters $\Sigma$ and $\Delta$ from part \eqref{Hubert} are respectively called the {\it limiting fractional sum of shifts} and the {\it limiting fractional difference of shifts}.
Note that our theorems are only sensitive to the values of $\Sigma$ and $\Delta$ taken modulo $1$, and so we often report them modulo $1$.
The parameter $\Lambda$ from the theorem is called the {\it limiting fractional length} since $\ell_\iota/p_\iota$ is the fractional length.

To maximize the crosscorrelation merit factor of quadratic character sequences, Lemma \ref{Christopher} shows that the optimal degree of truncation is exactly the same as for m-sequences: if we recall from Section \ref{Timothy} the number $\Lambdatrunc=0.557874\ldots$, which is the middle root of $2 x^3-6 x+3$, then we should truncate our quadratic character sequences so that their limiting fractional length is $\Lambdatrunc$.
We then arrange the shifts of our sequences to be such that (using the notation of Theorem \ref{Anthony}) we have $r_\iota/p_\iota \to (1-\Lambdatrunc)/2$ and $s_\iota/p_\iota=-\Lambdatrunc/2$ as $p_\iota \to \infty$, which makes $\Sigma=1/2-\Lambdatrunc$ and $\Delta=1/2$.
If we do all this, then we show (see Lemma \ref{Christopher} for details) that the asymptotic crosscorrelation demerit factor, as calculated in Theorem \ref{Anthony}\eqref{Hubert}, obtains a global minimum value of $0.299216\ldots$, the smallest root of $3 x^3-33 x^2+33 x-7$.  Or equivalently, the asymptotic crosscorrelation merit factor achieves a global maximum value of $3.342065\ldots$, the largest root of $7 x^3-33 x^2+ 33 x-3$.

The curves in Figure \ref{Percy} show the asymptotic values of autocorrelation and crosscorrelation demerit factors (calculated in Theorem \ref{Anthony}\eqref{Hubert}) as functions of limiting fractional sum of shifts $\Sigma$ for quadratic character sequences with limiting fractional length $\Lambdatrunc=0.557874\ldots$ and limiting fractional difference of shifts $\Delta=1/2$. 
The data points in the same figure show the autocorrelation and crosscorrelation demerit factors as functions of fractional sum of shifts for pairs of truncated, shifted Legendre sequences whose character is the quadratic character from the prime field $\F_{257}$.
The sequences have been truncated to length $143$ (to approximate fractional length $\Lambdatrunc$), and we run through all $257$ possible shifts of the first sequence, while always shifting the second sequence so as to make the difference in shifts between the first and second sequence $128$ (to approximate fractional difference of shifts $1/2$).
This process allows us to vary the fractional sum of shifts while keeping the fractional length and fractional difference of shifts constant.
(And recall that we always read the fractional sum of shifts modulo $1$.)
To avoid clutter we only plot the autocorrelation demerit factors for the first sequence of each pair: the autocorrelation demerit factor of the second sequence is about the same.  These data closely match the asymptotic values.
\begin{center}
\begin{figure}
\begin{center}
\caption{Autocorrelation (plusses, dashed curve) and crosscorrelation (dots, solid curve) demerit factors for pairs of shifted truncated Legendre sequences.  Curves are asymptotic values from Theorem \ref{Anthony} with $\Lambda=\Lambdatrunc$ and $\Delta=1/2$, points are data for Legendre sequences truncated from length $257$ to length $143$ with difference of shifts $128$.}\label{Percy}

\end{center}
\end{figure}
\end{center}

Note that both autocorrelation and crosscorrelation demerit factors are shift-dependent, and show opposite tendencies: when one is low, the other is high.  For instance, when the shifts are arranged to obtain the global maximum asymptotic crosscorrelation merit factor of $3.342065\ldots$ (demerit factor $0.299216\ldots$), the asymptotic autocorrelation merit factor becomes $0.796072\ldots$ (demerit factor 1.250167\ldots).
On the other hand, when the shifts are arranged to minimize autocorrelation, one obtains an autocorrelation merit factor of $3.601116\ldots$ (demerit factor $0.277691\ldots$), but then the crosscorrelation performance deteriorates to crosscorrelation merit factor $0.782661\ldots$ (demerit factor $1.277691\ldots$).
Between these extremes, we can achieve autocorrelation and crosscorrelation performance significantly better than randomly selected sequences (that is, we can make the autocorrelation demerit factors for both of our sequences and also their crosscorrelation demerit factors all smaller than $1$ simultaneously).
Regardless of whether we arrange the shifts to minimize autocorrelation or to minimize crosscorrelation, we obtain
\[
\CDF(f_\iota,g_\iota)+\sqrt{\DF(f_\iota) \DF(g_\iota)} \to 1.555383\ldots
\]
in the limit as the length of the sequences tends to infinity, which can be compared with the lower bound of Pursley-Sarwate in \eqref{Gary}.
If, on the other hand, we choose shifts that are halfway between these extremes, then we can obtain asymptotic autocorrelation merit factor $1.783394\ldots$ (demerit factor $0.560728\ldots$) and asymptotic crosscorrelation merit factor $1.717466\ldots$ (demerit factor $0.582252\ldots$) and thus
\begin{equation}\label{Donald}
\CDF(f_\iota,g_\iota)+\sqrt{\DF(f_\iota) \DF(g_\iota)} \to 1.1429812\ldots
\end{equation}
as the length of the sequences tends to infinity.

Unfortunately, the sequence pairs mentioned above that produce asymptotic crosscorrelation merit factor $3.342065\ldots$ are taken from slightly overlapping segments of the Legendre sequence (they have a limiting fractional difference of shifts of $1/2$, but their limiting fractional length is $\Lambdatrunc=0.557874\ldots$, which is slightly larger than $1/2$).
This will produce rather strong crosscorrelation values of $\approx 0.116 p$  at shifts $(p+1)/2$ and $(1-p)/2$, which is much stronger than the mean square value of $\approx 0.387 \sqrt{p}$ that the crosscorrelation merit factor implies.
We should also note that we can approach even closer to the Pursley-Sarwate bound (see \eqref{Gary}) than we observed in \eqref{Donald} if we allow pairs of sequences that include even larger overlapping portions of the original Legendre sequence,\footnote{For example, one can obtain $\CDF(f_\iota,g_\iota)+\sqrt{\DF(f_\iota) \DF(g_\iota)} \to 13/12$ as sequence length tends to infinity if $f_\iota$ and $g_\iota$ are Legendre sequences of natural length $p_\iota$ with fractional shifts tending to $1/8$ and $5/8$, respectively, as $p_\iota \to \infty$.  Unfortunately, the pair $(f_\iota,g_\iota)$ would have very large crosscorrelation peaks of magnitudes of $(p_\iota-1)/2$ and $(p_\iota+1)/2$ at shifts $(p_\iota+1)/2$ and $(1-p_\iota)/2$, respectively.} but this will only exacerbate the problem of having two very strong crosscorrelation peaks in spite of overall low mean square crosscorrelation.

If an application demands that there be no overlap, we can take a Legendre sequence of length $p_\iota$, shift $(p_\iota-1)/4$ places to the right, discard the last term, and obtain our pair $(f_\iota,g_\iota)$ by cutting what remains into two sequences of length $(p_\iota-1)/2$.  In terms of the notation of Theorem \ref{Anthony}, we have truncated to limiting fractional length $\Lambda=1/2$, and the shifts of our two sequences have $r_\iota/p_\iota \to 1/4$ and $s_\iota/p_\iota \to -1/4$, so that $\Sigma=0$ and $\Delta=1/2$.  This yields an asymptotic crosscorrelation merit factor of $3$, which is not quite as large as the maximum of $3.342065\ldots$ mentioned in the previous paragraph, but is still much higher than what one can achieve with m-sequences or nonquadratic character sequences of natural or modified length.

The curves in Figure \ref{Monica} show the asymptotic values of autocorrelation and crosscorrelation demerit factors as functions of limiting fractional sum of shifts for quadratic character sequences with limiting fractional length $1/2$ and limiting fractional difference of shifts $1/2$, as calculated in Theorem \ref{Anthony}\eqref{Hubert}.
The data points in the same figure show the autocorrelation and crosscorrelation demerit factors as functions of fractional sum of shifts for pairs of truncated, shifted Legendre sequences whose character is the quadratic character from the prime field $\F_{257}$.
The sequences have been truncated to length $128$ (to approximate fractional length $1/2$), and we run through all $257$ possible shifts of the first sequence, while always shifting the second sequence so as to make the difference in shifts between the first and second sequence $128$ (to approximate fractional difference of shifts $1/2$).  
This process allows us to vary the fractional sum of shifts while keeping the fractional length and fractional difference of shifts constant.
(And recall that we always read the fractional sum of shifts modulo $1$.)
To avoid clutter we only plot the autocorrelation demerit factors for the first sequence of each pair: the autocorrelation demerit factor of the second sequence is about the same.  These data closely match the asymptotic values.

As we saw in the previous figure, autocorrelation and crosscorrelation demerit factors are shift-dependent and show opposite tendencies: when one is low, the other is high.  One can achieve a crosscorrelation merit factor as high as $3$ (demerit factor $1/3$) with this construction, but then the autocorrelation merit factor becomes $3/4$ (demerit factor $4/3$).  And conversely, one can achieve an autocorrelation merit factor as high as $3$ (demerit factor $1/3$), but then the crosscorrelation merit factor drops to $3/4$ (demerit factor $4/3$).
Between these extremes, we can achieve autocorrelation and crosscorrelation performance significantly better than randomly selected sequences (that is, we can make the autocorrelation demerit factors for both of our sequences and also their crosscorrelation demerit factors all smaller than $1$ simultaneously).
Regardless of whether we arrange the shifts to minimize autocorrelation or to minimize crosscorrelation, we obtain
\[
\CDF(f_\iota,g_\iota)+\sqrt{\DF(f_\iota) \DF(g_\iota)} \to 5/3\ldots
\]
in the limit as the length of the sequences tends to infinity, which can be compared with the lower bound of Pursley-Sarwate in \eqref{Gary}.
If, on the other hand, we choose shifts that are halfway between these extremes, then we can obtain asymptotic autocorrelation merit factor $12/7$ (demerit factor $7/12$) and asymptotic crosscorrelation merit factor $12/7$ (demerit factor $7/12$) and thus
\[
\CDF(f_\iota,g_\iota)+\sqrt{\DF(f_\iota) \DF(g_\iota)} \to 7/6
\]
as the length of the sequences tends to infinity.
\begin{center}
\begin{figure}
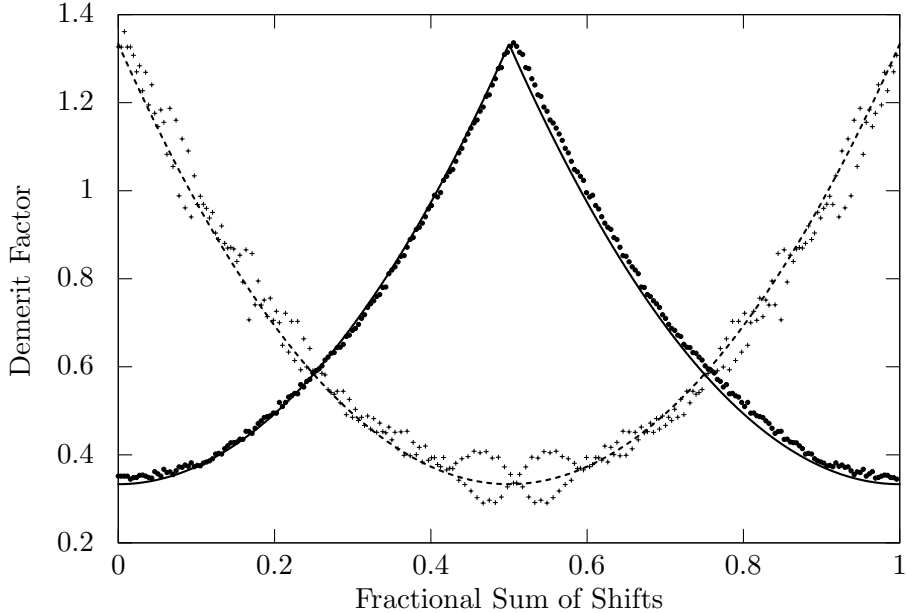

\begin{center}
\caption{Autocorrelation (plusses, dashed curve) and crosscorrelation (dots, solid curve) demerit factors for pairs of shifted truncated Legendre sequences.  Curves are asymptotic values from Theorem \ref{Anthony} with $\Lambda=1/2$ and $\Delta=1/2$, points are data for Legendre sequences truncated from length $257$ to length $128$ with difference of shifts $128$.}\label{Monica}

\end{center}
\end{figure}
\end{center}

The rest of the paper is devoted to proving the claims we have presented in Sections \ref{Abigail}--\ref{Milton}.  In Section \ref{Philip} we show the connection between merit factors and $L^p$ norms of polynomials on the complex unit circle: this provides a useful formalism for our proofs and shows that there is a deep connection between correlation problems and pure mathematics.
After some preliminary material in Section \ref{Gail} on Gauss sums (which are the Fourier coefficients for our sequences), we use Sections \ref{Simon}--\ref{Elizabeth} to furnish rigorous proofs of the results presented in Sections \ref{Abigail}--\ref{Milton}.
\FloatBarrier

\section{Connection with Polynomials and Norms}\label{Philip}

Recall that if $f=(f_0,\ldots,f_{\ell-1})$ and $g=(g_0,\ldots,g_{m-1})$ are sequences of complex numbers of lengths $\ell$ and $m$, and $s$ is an integer, then the (aperiodic) crosscorrelation of $f$ with $g$ at shift $s$ is defined to be
\begin{equation}\label{James}
C_{f,g}(s) = \sum_{j \in \Z} f_j \conj{g_{j+s}},
\end{equation}
where we use the convention that $f_j=0$ when $j\not\in\{0,1,\ldots,\ell-1\}$ and $g_k=0$ when $k \not \in \{0,1,\ldots,m-1\}$.
We can identify the sequence $f=(f_0,f_1,\ldots,f_{\ell-1})$ of complex numbers with the polynomial $f(z)=f_0+f_1 z+\cdots+f_{\ell-1} z^{\ell-1}$ in $\C[z]$.
We define $\conj{f}(z)=\conj{f_0}+\conj{f_1} z + \cdots + \conj{f_{\ell-1}} z^{\ell-1}$.
If $g=(g_0,\ldots,g_{m-1})$ is another sequence of complex numbers with associated polynomial $g(z)$, then one readily sees that the Laurent polynomial $f(z) \conj{g}(1/z)$ in $\C[z,z^{-1}]$ has $C_{f,g}(-s)$ for its coefficient of $z^s$.
As a special case, note that $C_{f,f}(0)=\sum_{s \in \Z} |f_s|^2$ is the constant coefficient of $f(z)\conj{f}(1/z)$.
By the same principle, we multiply $h(z)=f(z) \conj{g}(1/z)$ by $\conj{h}(1/z)=\conj{f}(1/z) g(z)$ to see that the constant coefficient of $f(z)\conj{f}(1/z) g(z) \conj{g}(1/z)$ is $\sum_{s\in\Z} |C_{f,g}(s)|^2$.

If $a(z)=\sum_{s \in \Z} a_s z^s$ is a Laurent polynomial (so all but finitely many $a_s$ are zero), integration around the complex unit circle extracts the constant coefficient
\[
a_0 = \frac{1}{2 \pi} \int_0^{2\pi} a(e^{i\theta}) d\theta.
\]
Furthermore, we note that $\conj{a}(1/z)=\conj{a(z)}$ if $z$ is on the complex unit circle, so our calculations in the previous paragraph show that
\[
C_{f,f}(0) = \frac{1}{2\pi} \int_0^{2\pi} |f(e^{i\theta})|^2 d\theta
\]
and
\[
\sum_{s \in \Z} |C_{f,g}(s)|^2 = \frac{1}{2 \pi} \int_0^{2 \pi} |f(e^{i \theta}) g(e^{i\theta})|^2 d\theta.
\]

For a real number $p \geq 1$, we define the $L^p$ norm of a complex-valued function $f$ on the complex unit circle as
\[
\norm{f}{p}=\left(\frac{1}{2\pi} \int_{0}^{2\pi} |f(e^{i\theta})|^p d\theta\right)^{1/p},
\]
so then the integrals in the previous paragraph can be expressed succinctly as $C_{f,f}(0) = \normtt{f}$ and $\sum_{s \in \Z} |C_{f,g}(s)|^2 = \normtt{f g}$.
Note that if $f=g$ (for autocorrelation calculations) then $\sum_{s \in \Z} |C_{f,f}(s)|^2=\normtt{f^2}=\normff{f}$.
Thus we can express the autocorrelation demerit factor \eqref{Adam} and merit factor \eqref{Amanda} as
\begin{align*}
\DF(f) & = \frac{\normff{f}}{\normtf{f}}-1 \\
\MF (f) & = \frac{\normtf{f}}{\normff{f}-\normtf{f}},
\end{align*}
and we can express the crosscorrelation demerit factor \eqref{Cedric} and merit factor \eqref{Carmine} as
\begin{align}
\CDF(f,g) & = \cnrd{f}{g} \label{Melissa} \\
\CMF(f,g) & = \frac{\normtt{f}\normtt{g}}{\normtt{f g}}. \nonumber
\end{align}
For the purposes of our mathematical analysis, it is actually most convenient to keep track of the ratio of norms $\cnr{f}{g}=\CDF(f,g)$, which becomes $\DF(f)+1$ when $f=g$ (i.e., in autocorrelation calculations).
The connection between the autocorrelation merit factor and the $L^4$ norm of polynomials on the complex unit circle furnishes a remarkable connection between research in pure mathematics going back to Littlewood (see \cite{Littlewood-1966} and \cite[Problem 19]{Littlewood-1968}) and in engineering going back to Golay (see \cite{Golay-1972}).

One can also obtain a useful summation formula for the sum of squared magnitudes of crosscorrelation values directly from \eqref{James}, as
\[
\sum_{s \in \Z} |C_{f,g}(s)|^2 = \sum_{s,j,k \in \Z} f_j \conj{g}_{j+s} \conj{f}_k g_{k+s}.
\]
which can be rewritten
\begin{equation}\label{Abraham}
\normtt{f g}=\sum_{s\in\Z} |C_{f,g}(s)|^2 = \sums{t,u,v,w \in \Z \\ t+u=v+w} f_t g_u \conj{f_v g_w}.
\end{equation}

\section{Characters and Gauss Sums}\label{Gail}

Fourier analysis is used in this work to study the crosscorrelation properties of additive and multiplicative characters of finite fields.  The Fourier coefficients for such characters are Gauss sums.
For the rest of this paper, $F$ is a finite field of characteristic $p$ and order $q$, and we write $\achars$ for the group of complex-valued additive characters of $F$ and $\mchars$ for the group of complex-valued multiplicative characters of $F$.
For any additive or multiplicative characters $\chi,\psi$, we write $\chi\psi$ for the product of characters $\chi\psi(x)=\chi(x)\psi(x)$, and when $n \in \Z$, we write $\chi^n$ for the function $\chi^n(x)=(\chi(x))^n$.  We write $\conj{\chi}$ for $\chi^{-1}$, so that $\chi^{-1}(x)=\conj{\chi(x)}$.
If $\chi\colon \Fu \to \C$ is any multiplicative character (including the trivial multiplicative character), we always extend $\chi$ to $F$ by setting $\chi(0)=0$.

Recall from Section \ref{Harold} that the canonical additive character of $F$ is the character $\epsilon \colon F \to \C$ with $\epsilon(x)=\exp(2\pi i \Tr_{F/\F_p}(x)/p)$, where $\Tr_{F/\F_p}$ is the absolute trace from Galois theory.
There are $|F|$ additive characters of $F$, given by the $|F|$ functions $x \mapsto \epsilon(a x)$ for each $a \in F$, where $a=0$ yields the trivial character.
We write $\epsilon_a$ for the additive character $x\mapsto \epsilon(a x)$.

We have the following {\it orthogonality relations} for characters.
\begin{lemma}\label{Orestes}
\mbox{}
\begin{enumerate}[(i).]
\item If $\psi \in \achars$, then $\sum_{a \in F} \psi(a)=\begin{cases} \card{F} & \text{if $\psi$ is trivial}, \\ 0 & \text{otherwise}.\end{cases}$
\item\label{Sarah} If $a \in F$, then $\sum_{\psi\in\achars} \psi(a)=\begin{cases} \card{F} & \text{if $a=0$}, \\ 0 & \text{otherwise}.\end{cases}$
\item\label{Theodore} If $\chi \in \mchars$, then $\sum_{b \in \Fu} \chi(b)=\begin{cases} \card{\Fu} & \text{if $\chi$ is trivial}, \\ 0 & \text{otherwise}.\end{cases}$
\item\label{Edith} If $b \in \Fu$, then $\sum_{\chi\in\mchars} \chi(b)=\begin{cases} \card{\Fu} & \text{if $b=1$}, \\ 0 & \text{otherwise}.\end{cases}$
\end{enumerate}
\end{lemma}
\begin{proof}
All the cases where the sum comes to zero can be found in \cite[Theorem 5.4]{Lidl-Niederreiter}), and all the cases where the sum comes out to be a positive number are trivial.
\end{proof}

The {\it Gauss sum} associated with an additive character $\psi \in \achars$ and a multiplicative character $\chi \in \mchars$, is
\[
G(\psi,\chi) = \sum_{x \in \Fu} \psi(x) \chi(x).
\]
When $\psi$ is expressed as $\epsilon_a$ for some $a \in F$, we have the shorthand
\begin{equation}\label{George}
G_a(\chi) = G(\epsilon_a,\chi) = \sum_{x \in \Fu} \epsilon_a(x) \chi(x),
\end{equation}
and when $\psi=\epsilon=\epsilon_1$, the canonical additive character, we have the shorthand
\begin{equation}\label{Gerald}
G(\chi)=G_1(\chi)=\sum_{x \in \Fu} \epsilon(x) \chi(x).
\end{equation}
We summarize the facts that we shall need about Gauss sums.
\begin{lemma}\label{Gabriel}
If $a \in F$ and $\chi \in \mchars$, then
\begin{enumerate}[(i).]
\item\label{Henry} $G_0(\chi)=\card{\Fu}$ if $\chi$ is trivial,
\item\label{Irene} $G_0(\chi)=0$ if $\chi$ is nontrivial,
\item\label{Jerry} $G_a(\chi)=G(\chi)=-1$ if $a\in \Fu$ and $\chi$ is trivial,
\item\label{Katherine} $|G_a(\chi)|=|G(\chi)|=\sqrt{\card{F}}$ if $a\in \Fu$ and $\chi$ is nontrivial,
\item\label{Lawrence} $G_a(\chi)=\conj{\chi}(a) G(\chi)$ if $a\in \Fu$ or $\chi$ is nontrivial,
\item\label{Manuel} $G_a(\conj{\chi})=\chi(-1) \conj{G_a(\chi)}$, so that $G(\conj{\chi})=\chi(-1) \conj{G(\chi)}$, and
\item\label{Nelson} if $F$ is of characteristic $p$ and order $q$ and $d \in \Z$ with $d \equiv p^k \pmod{q-1}$, then $G(\chi^d)=G(\chi)$ and $G(\chi^{-d})=\chi(-1) \conj{G(\chi)}$.
\end{enumerate}
\end{lemma}
\begin{proof}
All of these are proved in \cite[Theorems 5.11, 5.12]{Lidl-Niederreiter}, except the case where $a=0$ in \eqref{Lawrence}, which follows from \eqref{Irene}, and except case \eqref{Nelson}, which follows easily from \cite[Theorem 5.12(v)]{Lidl-Niederreiter} and \eqref{Manuel}.
\end{proof}

The following lemma, proved in \cite[eq.~(5.17)]{Lidl-Niederreiter}, gives the Fourier expansion of additive characters in terms of multiplicative characters.
\begin{lemma}\label{Stephen}
For $a \in F$ and $b \in \Fu$ we have
\[
\epsilon_a(b) = \frac{1}{\card{\Fu}} \sum_{\xi\in\mchars} G_a(\xi) \conj{\xi}(b).
\]
\end{lemma}
And we also have the Fourier expansion of multiplicative characters in terms of additive characters:
\begin{lemma}\label{Theresa}
For $\chi \in \mchars$ and $b \in F$, we have
\[
\chi(b) = \frac{1}{\card{F}} \sum_{a \in F} G_a(\chi) \conj{\epsilon_a}(b).
\]
\end{lemma}
\begin{proof}
The case where $b\not=0$ is proved in \cite[eq.~(5.16)]{Lidl-Niederreiter}, and the case where $b=0$ is tantamount to proving that $\sum_{a \in F} G_a(\chi)=0$, which follows from Lemma \ref{Gabriel}\eqref{Henry},\eqref{Jerry} if $\chi$ is trivial.  If $\chi$ is nontrivial, then Lemma \ref{Gabriel}\eqref{Lawrence} shows that $\sum_{a \in F} G_a(\chi)=G(\chi) \sum_{a \in F} \conj{\chi}(a)$, which vanishes by Lemma \ref{Orestes}\eqref{Theodore}.
\end{proof}

\section{Proof of Theorem \ref{Victor}}\label{Simon}

In this section, we prove Theorem \ref{Victor} (of which Theorems \ref{Thomas} and \ref{Deidre} are special cases).
Let us recall the basic assumptions in Theorem \ref{Victor}:
\begin{itemize}
\item We have a fixed prime $p$ and a fixed integer $d$.
\item We have an infinite family of pairs of sequences $\{(f_\iota,g_\iota)\}_{\iota \in I}$, where for each $\iota \in I$, the sequence $f_\iota$ is the $p$-ary additive character sequence associated to primitive element $\alpha_\iota$ of field $F_\iota$ with shift $r_\iota$ and length $\ell_\iota$, while $g_\iota$ is the $p$-ary additive character sequence associated to primitive element $\alpha_\iota^d$ of field $\Fi$ with shift $s_\iota$ and length $\ell_\iota$.
\item For every $\iota \in I$, we have $q_\iota=\card{\Fi}$ with $\gcd(d,q_\iota-1)=1$.
\item The set $\{q_\iota\}_{\iota \in I}$ is unbounded, and there is a positive real number $\Lambda$ such that in the limit as $q_\iota \to \infty$, we have $\ell_\iota/(q_\iota-1) \to \Lambda$.
\end{itemize}
In certain subcases of Theorem \ref{Victor}, there are additional assumptions, which will be recalled when we address these cases in turn.

Let $(f,g)$ be a pair of additive character sequences from our family, where $f$ is associated to primitive element $\alpha$ of field $F$ with shift $r$ and length $\ell$, and $g$ is associated to primitive element $\alpha^d$ of field $F$ with shift $s$ and length $\ell$.  Let $q=\card{F}$, and note that $\gcd(d,q-1)=1$.  Let $L=\{0,1,\ldots,\ell-1\}$.
So then $f=(f_0,\ldots,f_{\ell-1})$ and $g=(g_0,\ldots,g_{\ell-1})$ with
\begin{align*}
f_j & = \epsilon(\alpha^{j+r}) \\
g_j & = \epsilon(\alpha^{d(j+s)})
\end{align*}
for every $j \in L$.

We begin by calculating $\normtt{f g}$, the numerator of $\CDF(f,g)$ in \eqref{Melissa}.
From \eqref{Abraham} we have
\[
\normtt{f g}  = \sums{t,u,v,w \in L \\ t+u=v+w} \epsilon(\alpha^{t+r}) \epsilon(\alpha^{d(u+s)}) \conj{\epsilon(\alpha^{v+r}) \epsilon(\alpha^{d(w+s)})}.
\]
We now use Lemma \ref{Stephen} to express the additive character $\epsilon$ in terms of multiplicative characters to see that $\normtt{f g}$ is
\[
\sums{t,u,v,w \in L \\t+u=v+w} \sums{\kappa,\lambda \in \mchars \\ \mu,\nu \in \mchars} \frac{G(\kappa) G(\lambda) \conj{G(\mu) G(\nu)}}{(q-1)^4} \conj{\kappa(\alpha^{t+r}) \lambda(\alpha^{d(u+s)})} \mu(\alpha^{v+r}) \nu(\alpha^{d(w+s)}).
\]
For a given $\xi \in \mchars$, the quadruple $(\kappa \xi^d,\lambda \xi, \mu \xi^d,\nu\xi)$ runs through $(\mchars)^4$ as the quadruple $(\kappa,\lambda,\mu,\nu)$ does, so the inner sum is equal to
\[
\sums{\kappa,\lambda \in \mchars \\ \mu,\nu \in \mchars} \! \frac{G(\kappa\xi^d) G(\lambda\xi) \conj{G(\mu\xi^d) G(\nu\xi)}}{(q-1)^4} \conj{\kappa \xi^d(\alpha^{t+r}) \lambda \xi (\alpha^{d(u+s)})} \mu\xi^d (\alpha^{v+r}) \nu\xi(\alpha^{d(w+s)})
\]
for every $\xi \in \mchars$, and so we can average over all $\xi$ and see that the inner sum is equal to
\[
\sums{\kappa,\lambda,\mu,\nu,\xi \in \mchars} \frac{G(\kappa\xi^d) G(\lambda\xi) \conj{G(\mu\xi^d) G(\nu\xi)}}{(q-1)^5} \conj{\kappa(\alpha^{t+r}) \lambda (\alpha^{d(u+s)})} \mu(\alpha^{v+r}) \nu(\alpha^{d(w+s)}),
\]
where we have used the fact that $\conj{\xi^d(\alpha^{t+r}) \xi(\alpha^{d(u+s)})} \xi^d(\alpha^{v+r}) \xi(\alpha^{d(w+s)})$ is always $1$ because we always have $t+u=v+w$ in the outer sum.  We therefore define
\[
H(\kappa,\lambda,\mu,\nu)=\frac{1}{(q-1)^3} \sum_{\xi \in \mchars} G(\kappa\xi^d) G(\lambda\xi) \conj{G(\mu\xi^d) G(\nu \xi)},
\]
and then
\begin{equation}\label{Herbert}
\normtt{f g} = \!\!\!\! \sums{\kappa,\lambda \in \mchars \\ \mu,\nu \in \mchars} \!\! \frac{H(\kappa,\lambda,\mu,\nu)}{(q-1)^2} \!\!\!\! \sums{t,u,v,w \in L \\ t+u=v+w} \!\!\!\! \conj{\kappa(\alpha^{t+r}) \lambda(\alpha^{d(u+s)})} \mu(\alpha^{v+r}) \nu(\alpha^{d(w+s)}).
\end{equation}

For any $a\in \Z$, let $\delta_{a\sim 1}=1$ if $a=p^k$ for some integer $k$, and $\delta_{a\sim 1}=0$ otherwise.
Since we are computing a limit as $q \to \infty$, we may take $q-1 > |d|$, so then $d$ is congruent to a power of $p$ modulo $q-1$ if and only if $d$ is equal to a power of $p$, and likewise $-d$ is congruent to a power of $p$ modulo $q-1$ if and only if $-d$ is a power of $p$.
Then by Lemma \ref{Edward} below, we may write $H(\kappa,\lambda,\mu,\nu)=M(\kappa,\lambda,\mu,\nu)+E(\kappa,\lambda,\mu,\nu)$, where
\[
M(\kappa,\lambda,\mu,\nu) = \begin{cases}
1 & \text{if $(\kappa,\lambda)=(\mu,\nu)$,} \\
\delta_{d\sim 1} & \text{if $(\kappa,\lambda)\not=(\mu,\nu)$ and $(\kappa,\mu)=(\nu^d,\lambda^d)$,} \\
\delta_{-d\sim 1} \kappa\conj{\mu}(-1) & \text{if $(\kappa,\lambda)\not=(\mu,\nu)$ and $(\kappa,\mu)=(\lambda^d,\nu^d)$,} \\
0 & \text{otherwise,}
\end{cases}
\]
is regarded as the main term, and $E(\kappa,\lambda,\mu,\nu)$ is regarded as the error term with
\[
|E(\kappa,\lambda,\mu,\nu)| \leq |d| q^{3/2}/(q-1)^2.
\]

So, proceeding from \eqref{Herbert}, we see that $\normtt{f g}=M_0+E_0$, where
\begin{align*}
M_0 & = \sums{\kappa,\lambda \in \mchars \\ \mu,\nu \in \mchars} \frac{M(\kappa,\lambda,\mu,\nu)}{(q-1)^2} \sums{t,u,v,w \in L \\ t+u=v+w} \conj{\kappa(\alpha^{t+r}) \lambda(\alpha^{d(u+s)})} \mu(\alpha^{v+r}) \nu(\alpha^{d(w+s)}) \\
E_0 & = \sums{\kappa,\lambda \in \mchars \\ \mu,\nu \in \mchars} \frac{E(\kappa,\lambda,\mu,\nu)}{(q-1)^2} \sums{t,u,v,w \in L \\ t+u=v+w} \conj{\kappa(\alpha^{t+r}) \lambda(\alpha^{d(u+s)})} \mu(\alpha^{v+r}) \nu(\alpha^{d(w+s)}),
\end{align*}
and our bound on $E(\kappa,\lambda,\mu,\nu)$ tells us that
\[
|E_0| \leq \frac{|d| q^{3/2}}{(q-1)^4} \sums{\kappa,\lambda \in \mchars \\ \mu,\nu \in \mchars} \left|\sums{t,u,v,w \in L \\ t+u=v+w} \conj{\kappa(\alpha^{t+r}) \lambda(\alpha^{d(u+s)})} \mu(\alpha^{v+r}) \nu(\alpha^{d(w+s)})\right|,
\]
and so
\begin{equation}\label{Edwin}
|E_0| \leq 64 |d| q^{3/2} (1+\log(q-1))^3 \max\left(1,\frac{\ell}{q-1}\right)^3
\end{equation}
by Lemma \ref{Oliver} below.

The main term $M_0$ can be broken into four terms, $M_1$, $M_2$, $M_3$, and $M_4$, based on the four mutually exclusive cases where $M(\kappa,\lambda,\mu,\nu)$ can be nonzero: (1) $\kappa=\mu=\lambda^d=\nu^d$, (2) $\kappa=\mu\not=\lambda^d=\nu^d$, (3) $\kappa=\nu^d\not=\mu=\lambda^d$, and (4) $\kappa=\lambda^d\not=\mu=\nu^d$.
Then
\[
\normtt{f g} = M_1+M_2+M_3+M_4+E_0,
\]
where
\begin{align*}
M_1 & = \frac{1}{(q-1)^2} \sum_{\kappa \in \mchars} \sums{t,u,v,w \in L \\ t+u=v+w} 1 \\
M_2 & = \frac{1}{(q-1)^2} \sums{\kappa,\lambda, \in \mchars \\ \kappa\not=\lambda^d} \sums{t,u,v,w \in L \\ t+u=v+w} \kappa(\alpha^{v-t}) \lambda(\alpha^{d(w-u)}) \nonumber \\
M_3 & = \frac{\delta_{d\sim 1}}{(q-1)^2} \sums{\kappa,\mu \in \mchars \\ \kappa\not=\mu} \sums{t,u,v,w \in L \\ t+u=v+w} \mu(\alpha^{v-u+r-s}) \kappa(\alpha^{w-t+s-r}) \\
M_4 & = \frac{\delta_{d\sim -1}}{(q-1)^2} \sums{\kappa,\mu \in \mchars \\ \kappa\not=\mu} \sums{t,u,v,w \in L \\ t+u=v+w} \kappa\conj{\mu}(-1) \conj{\kappa(\alpha^{t+u+r+s})} \mu(\alpha^{v+w+r+s}).
\end{align*}
We let $\sigma=0$ if the characteristic $p$ of our field is $2$, and $\sigma=(q-1)/2$ if $p$ is odd, so that $-1=\alpha^\sigma=\alpha^{-\sigma}$ in our field, and so
\[
M_4 = \frac{\delta_{d\sim -1}}{(q-1)^2} \sums{\kappa,\mu \in \mchars \\ \kappa\not=\mu} \sums{t,u,v,w \in L \\ t+u=v+w} \conj{\kappa(\alpha^{t+u+r+s+\sigma})} \mu(\alpha^{v+w+r+s+\sigma}).
\]
Let 
\begin{align*}
A & = \{(t,u,v,w) \in L^4: t+u=v+w\} \\
B & = \{(t,u,v,w) \in A: v-t \equiv 0 \!\! \pmod{q-1}\} \\
C & = \{(t,u,v,w) \in A: w-t \equiv r-s \!\! \pmod{q-1}\} \\
D & = \{(t,u,v,w) \in A: t+u \equiv -(r+s+\sigma) \!\! \pmod{q-1}\},
\end{align*}
so that, by the orthogonality relation in Lemma \ref{Orestes}\eqref{Edith}, we have
\begin{align*}
M_1 & = \frac{|A|}{q-1} \\
M_2 & = |B|-\frac{|A|}{q-1} \\
M_3 & = \delta_{d\sim 1} \left(|C|-\frac{|A|}{q-1}\right) \\
M_4 & = \delta_{-d\sim 1} \left(|D|-\frac{|A|}{q-1}\right),
\end{align*}
and so $M_0=M_1+M_2+M_3+M_4=M_A+M_B+M_C+M_D$, where
\begin{align*}
M_A & = -\frac{(\delta_{d\sim 1}+\delta_{-d\sim 1})|A|}{q-1} \\
M_B & = |B| \\
M_C & = \delta_{d\sim 1} |C| \\
M_D & = \delta_{-d\sim 1} |D|,
\end{align*}
Recall the function $\Omega$ defined in \eqref{Orlando}.
Using the cardinalities of $A$, $B$, $C$, and $D$ computed in Lemmata \ref{Iago}--\ref{Leslie} below, we obtain $\normtt{f g}=M_0+E_0=M_A+M_B+M_C+M_D+E_0$ with
\begin{align*}
M_A & = -\frac{(\delta_{d\sim 1}+\delta_{-d\sim 1})(2\ell^3+\ell)}{3(q-1)} \\
M_B & = \ell^2 \Omega\left(\frac{q-1}{\ell},0\right) \\
M_C & = \delta_{d\sim 1} \ell^2 \Omega\left(\frac{q-1}{\ell},\frac{r-s}{\ell}\right) \\
M_D & = \delta_{-d\sim 1} \ell^2 \Omega\left(\frac{q-1}{\ell},1+\frac{r+s+\sigma-1}{\ell}\right),
\end{align*}
and where $E_0$ is bounded in \eqref{Edwin}.

Now $\normtt{f}=\normtt{g}=\ell$ for our sequences since they are of length $\ell$ with unimodular terms, so $\cnr{f}{g}=\normtt{f g}/\ell^2$.
Thus
\[
\cnrd{f}{g}=N_A+N_B+N_C+N_D+E_1
\]
with
\begin{align*}
N_A & = -\frac{(\delta_{d\sim 1}+\delta_{-d\sim 1})(2\ell^2+1)}{3(q-1)\ell} \\
N_B & = \Omega\left(\frac{q-1}{\ell},0\right) \\
N_C & = \delta_{d\sim 1} \Omega\left(\frac{q-1}{\ell},\frac{r-s}{\ell}\right) \\
N_D & = \delta_{-d\sim 1} \Omega\left(\frac{q-1}{\ell},1+\frac{r+s+\sigma-1}{\ell}\right) \\
|E_1| & \leq \frac{64 |d| q^{3/2} (1+\log(q-1))^3}{\ell^2} \max\left(1,\frac{\ell}{q-1}\right)^3.
\end{align*}
Now that we have estimated $\CDF(f,g)=\cnr{f}{g}$ for a single pair $(f,g)$ from the family $\{(f_\iota,g_\iota)\}_{\iota \in I}$, let us consider the asymptotic behavior of $\cnr{f}{g}$ as $q \to \infty$.
The ratio $\ell/(q-1)$ tends to the positive real number $\Lambda$, so that $N_A \to -2 (\delta_{d\sim 1}+\delta_{-d\sim 1}) \Lambda/3$.  Note that this limiting value of $N_A$ is zero unless $|d|$ is a power of $p$, in which case it becomes $-2\Lambda/3$.
Recall that $\Omega$ is continuous on $\{(x,y): x\not=0\}$, so that $N_B \to \Omega(1/\Lambda,0)$ as $q\to \infty$.
The term $N_C$ vanishes unless $d$ is a power of $p$, in which case we are given the extra assumption that $(r-s)/(q-1) \to \Delta$ as $q\to\infty$, and then $N_C \to \Omega(1/\Lambda,\Delta/\Lambda)$ as $q\to\infty$.
The term $N_D$ vanishes unless $-d$ is a power of $p$, in which case we are given the extra assumption that $(r+s)/(q-1) \to \Sigma$ as $q\to\infty$, and then $N_D \to \Omega(1/\Lambda,1+\Sigma^\prime/\Lambda)$ as $q\to\infty$, where $\Sigma^\prime=\Sigma$ in characteristic $2$ (where $\sigma=0$) and $\Sigma^\prime=\Sigma+1/2$ in odd characteristic (where $\sigma=(q-1)/2$).
The term $E_1 \to 0$ as $q\to\infty$ because $\ell/(q-1)$ tends to a positive real number $\Lambda$ as $q \to\infty$.
This completes the proof of Theorem \ref{Victor}.\hfill\qedsymbol

We conclude this section with the technical lemmata employed in the above proof of Theorem \ref{Victor}.
\begin{lemma}\label{Edward}
Let $F$ be the finite field of order $q$, and let $\kappa,\lambda,\mu,\nu \in \mchars$.  Let $d$ be an integer with $\gcd(d,q-1)=1$.  Let
\[
H=
\frac{1}{(q-1)^3} \sum_{\xi \in \mchars} G(\kappa\xi^d) G(\lambda\xi) \conj{G(\mu\xi^d) G(\nu \xi)}.
\]
\begin{enumerate}[(i).]
\item\label{Patricia} If $d\equiv p^k \pmod{q-1}$ for some $k \in \Z$, then
\[
H = \begin{cases}
1 + E & \text{if $\kappa=\mu$ and $\lambda=\nu$,} \\
1 + E & \text{if $\kappa=\nu^d$ and $\mu=\lambda^d$,} \\
E & \text{otherwise,}
\end{cases}
\]
for some $E \in \C$ with $|E| \leq q^{3/2}/(q-1)^2$.
\item\label{Rebecca} If $d\equiv -p^k \pmod{q-1}$ for some $k \in \Z$, then
\[
H = \begin{cases}
1 + E & \text{if $\kappa=\mu$ and $\lambda=\nu$,} \\
\kappa\conj{\mu}(-1) + E & \text{if $\kappa=\lambda^d$ and $\mu=\nu^d$,} \\
E & \text{otherwise,}
\end{cases}
\]
for some $E \in \C$ with $|E| \leq q^{3/2}/(q-1)^2$.
\item\label{Oscar} If $d\not\equiv p^k \pmod{q-1}$ and $d\not\equiv -p^k \pmod{q-1}$ for every $k \in \Z$, then
\[
H= \begin{cases}
1+E & \text{if $\kappa=\mu$ and $\lambda=\nu$,} \\
E & \text{otherwise,}
\end{cases}
\]
for some $E \in \C$ with $|E| \leq |d| q^{3/2}/(q-1)^2$.
\end{enumerate}
\end{lemma}
\begin{proof}
We write out the Gauss sums according to their definition \eqref{Gerald} to see that
\[
H=\frac{1}{(q-1)^3} \sum_{\xi \in \mchars} \sum_{w,x,y,z \in \Fu} \epsilon(w+x-y-z) \kappa\xi^d(w) \lambda\xi(x) \conj{\mu\xi^d(y) \nu\xi(z)},
\]
and when we sum $\xi^d(w) \xi(x) \conj{\xi^d(y) \xi(z)}$ over $\xi$, Lemma \ref{Orestes}\eqref{Edith} tells us that we obtain $q-1$ for terms with $w^d x=y^d z$, and zero for all other terms, so 
\[
H=\frac{1}{(q-1)^2} \sums{w,x,y,z \in \Fu \\ w^d x = y^d z} \epsilon(w+x-y-z) \kappa(w) \lambda(x) \conj{\mu(y) \nu(z)}.
\]
Now we reparameterize with $w=u y$ and $z=u^d x$ to obtain
\begin{align*}
H & = \frac{1}{(q-1)^2} \sum_{u,x,y \in \Fu} \epsilon((u-1)y+(1-u^d) x) \kappa\conj{\nu}^d(u) \kappa\conj{\mu}(y) \lambda\conj{\nu}(x) \\
& = \frac{1}{(q-1)^2} \sum_{u \in \Fu} G_{u-1}(\kappa\conj{\mu}) G_{1-u^d}(\lambda\conj{\nu}) \kappa\conj{\nu}^d(u),
\end{align*}
where we have used the definition of Gauss sums \eqref{George} in the second equality.
When we use Lemma \ref{Gabriel}\eqref{Henry},\eqref{Irene},\eqref{Lawrence}, and let $\delta$ be the Kronecker delta, we obtain
\begin{equation}\label{Ronald}
H = \delta_{\kappa,\mu} \delta_{\lambda,\nu} + \frac{G(\kappa\conj{\mu}) G(\lambda\conj{\nu})}{(q-1)^2} \sums{u \in \Fu \\u\not=1} \conj{\kappa}\mu(u-1) \conj{\lambda}\nu(1-u^d) \kappa\conj{\nu}^d(u).
\end{equation}

To prove case \eqref{Oscar}, where $|d|$ is not a power of $p$ modulo $q-1$, we finish by bounding the magnitude of the second term of \eqref{Ronald}, namely,
\[
S=\frac{G(\kappa\conj{\mu}) G(\lambda\conj{\nu})}{(q-1)^2} \sums{u \in \Fu \\u\not=1} \conj{\kappa}\mu(u-1) \conj{\lambda}\nu(1-u^d) \kappa\conj{\nu}^d(u).
\]
Note that $S=0$ if $q=2$, so we assume that $q>2$ henceforth, and then the condition $\gcd(d,q-1)=1$ forces $d\not=0$.
If $d < 0$, the second term of \eqref{Ronald} can be rewritten as
\[
S=\frac{G(\kappa\conj{\mu}) G(\lambda\conj{\nu})}{(q-1)^2} \sums{u \in \Fu \\u\not=1} \conj{\kappa}\mu(u-1) \conj{\lambda}\nu(u^{-d}-1) \kappa\conj{\lambda}^d(u).
\]
Then note that the polynomial $1-u^d$ (if $d >0$) or $u^{-d}-1$ (if $d < 0$) has at most $|d|$ roots in the algebraic closure of $F$, one of which is $1$, but the rest are not roots of $u$ or $u-1$.  Thus the Weil bound \cite{Weil} shows that the sum over $u$ in $S$ is at most $|d| \sqrt{q}$ unless $\lambda=\nu$, $\kappa=\mu$, and $\kappa=\nu^d$, in which case the sum over $u$ is $q-2$.  If the sum over $u$ is no greater than $|d|\sqrt{q}$, we use Lemma \ref{Gabriel}\eqref{Jerry},\eqref{Katherine} to see that $|S|\leq |d| q^{3/2}/(q-1)^2$, but if $\lambda=\nu$ and $\kappa=\mu$, then $G(\kappa\conj{\mu}) G(\lambda\conj{\nu})=1$ by Lemma \ref{Gabriel}\eqref{Jerry}, so that $|S| \leq (q-2)/(q-1)^2$.

In case \eqref{Patricia}, we have $d \equiv p^k \pmod{q-1}$ for some $k \in \Z$, so that $(1-u^d)=(1-u)^d$ for $u\in \Fu$ with $u\not=1$.  Thus \eqref{Ronald} becomes
\begin{equation}\label{Jason}
H = \delta_{\kappa,\mu} \delta_{\lambda,\nu} + \frac{\conj{\lambda}\nu(-1) G(\kappa\conj{\mu}) G(\lambda\conj{\nu})}{(q-1)^2} \sums{u \in \Fu \\u\not=1} \conj{\kappa}\mu\conj{\lambda}^d\nu^d(u-1) \kappa\conj{\nu}^d(u),
\end{equation}
and by the Weil bound \cite{Weil}, the sum over $u$ has magnitude no greater than $\sqrt{q}$ unless $(\kappa,\mu)=(\nu^d,\lambda^d)$.  This and Lemma \ref{Gabriel}\eqref{Jerry},\eqref{Katherine} ensure that $|H| \leq q^{3/2}/(q-1)^2$ when we have neither $(\kappa,\lambda)=(\mu,\nu)$ nor $(\kappa,\mu)=(\nu^d,\lambda^d)$.  If $(\kappa,\lambda)=(\mu,\nu)$, then Lemma \ref{Gabriel}\eqref{Jerry} and a trivial bound on the sum over $u$ in \eqref{Jason} show that $H=1+E$ for some $E$ with $|E| \leq (q-2)/(q-1)^2$.  If $(\kappa,\mu)=(\nu^d,\lambda^d)$ but $(\kappa,\lambda)\not=(\mu,\nu)$, then Lemma \ref{Gabriel}\eqref{Nelson},\eqref{Katherine} shows that the product of Gauss sums in \eqref{Jason} is $\lambda\conj{\nu}(-1) q$, so that we have $H=q(q-2)/(q-1)^2=1-1/(q-1)^2$.

In case \eqref{Rebecca}, we have $d \equiv -p^k \pmod{q-1}$ for some $k \in \Z$, so that $(1-u^d)=(1-u^{-1})^{-d}$ for $u \in \Fu$ with $u\not=1$.
Thus \eqref{Ronald} becomes
\begin{equation}\label{Edgar}
H = \delta_{\kappa,\mu} \delta_{\lambda,\nu} + \frac{G(\kappa\conj{\mu}) G(\lambda\conj{\nu})}{(q-1)^2} \sums{u \in \Fu \\u\not=1} \conj{\kappa}\lambda^d\mu\conj{\nu}^d(u-1) \kappa\conj{\lambda}^d(u),
\end{equation}
and by the Weil bound \cite{Weil}, the sum over $u$ has magnitude no greater than $\sqrt{q}$ unless $(\kappa,\mu)=(\lambda^d,\nu^d)$.  This and Lemma \ref{Gabriel}\eqref{Jerry},\eqref{Katherine} ensure that $|H| \leq q^{3/2}/(q-1)^2$ when we have neither $(\kappa,\lambda)=(\mu,\nu)$ nor $(\kappa,\mu)=(\lambda^d,\nu^d)$.  If $(\kappa,\lambda)=(\mu,\nu)$, then Lemma \ref{Gabriel}\eqref{Jerry} and a trivial bound on the sum over $u$ in \eqref{Edgar} show that $H=1+E$ for some $E$ with $|E| \leq (q-2)/(q-1)^2$.  If $(\kappa,\mu)=(\lambda^d,\nu^d)$ but $(\kappa,\lambda)\not=(\mu,\nu)$, then Lemma \ref{Gabriel}\eqref{Nelson},\eqref{Katherine} shows that the product of Gauss sums in \eqref{Edgar} is $\lambda\conj{\nu}(-1) q$, so that we have $H=\lambda\conj{\nu}(-1) q(q-2)/(q-1)^2=\lambda\conj{\nu}(-1) [1-1/(q-1)^2]$.  We have $(-1)^d=-1$ (trivially in characteristic $2$, but also in odd characteristic since $\gcd(d,q-1)=1$ forces $d$ to be odd), so that $\lambda\conj{\nu}(-1)=\lambda\conj{\nu}((-1)^d)=\kappa\conj{\mu}(-1)$ when $(\kappa,\mu)=(\lambda^d,\nu^d)$.
\end{proof}

\begin{lemma}\label{Oliver}
Let $F$ be the finite field of order $q$ and let $\alpha$ be a primitive element of $F$.  Let $a,b,c,d, h,i,j,k,\ell \in \Z$ with $\gcd(a,q-1)=\gcd(b,q-1)=\gcd(c,q-1)=\gcd(d,q-1)=1$ and $\ell>0$, and let $L=\{0,1,\ldots,\ell-1\}$.  Then the sum
\[
S=\sums{\kappa,\lambda \in \mchars \\ \mu,\nu \in \mchars} \left|\sums{t,u,v,w \in L \\ t+u=v+w} \conj{\kappa(\alpha^{a(t+h)}) \lambda(\alpha^{b(u+i)})} \mu(\alpha^{c(v+j)}) \nu(\alpha^{d(w+k)})\right|
\]
satisfies $|S| \leq 64 (q-1)^4 (1+\log(q-1))^3 \max(1,\ell/(q-1))^3$.
\end{lemma}
\begin{proof}
We have
\begin{align*}
S
& = \sums{\kappa,\lambda \in \mchars \\ \mu,\nu \in \mchars} \left|\sums{t,u,v,w \in L \\ t+u=v+w} \conj{\kappa^a(\alpha^t) \lambda^b(\alpha^u)} \mu^c(\alpha^v) \nu^d(\alpha^w)\right| \\
& = \sums{\kappa,\lambda \in \mchars \\ \mu,\nu \in \mchars} \left|\sums{t,u,v,w \in L \\ t+u=v+w} \conj{\kappa(\alpha^t) \lambda(\alpha^u)} \mu(\alpha^v) \nu(\alpha^w)\right|,
\end{align*}
where the first equality follows because $\conj{\kappa(\alpha^{a h})\lambda(\alpha^{b i})}\mu(\alpha^{c j})\nu(\alpha^{d k})$ is just a unimodular complex number, and the second equality follows because $\kappa^a$ runs through $\mchars$ as $\kappa$ runs through $\mchars$ since $\gcd(a,q-1)=1$, and similarly with $\lambda^b$, $\mu^c$, and $\nu^d$.
Now we employ \cite[Lemma A.1]{Katz-2013} (using $\Gamma=\Fu$, $n=1$, $\pi_1(j)=\alpha^j$, $U=L=\{0,1,\ldots,\ell-1\}$) to see that the last expression for $S$ is bounded above by $64 (q-1)^4 (1+\log(q-1))^3 \max(1,\ell/(q-1))^3$.
\end{proof}

\begin{lemma}\label{Iago}
Let $S$ be a finite set of consecutive integers.
Let $I$ be the set $\{(t,u,v,w) \in S^4: t+u=v+w\}$.
Then $\card{I}=\frac{2 |S|^3+|S|}{6}$.
\end{lemma}
\begin{proof}
For each $n\in\Z$, the number of pairs $(t,v) \in S^2$ that make $t-v=n$ is $\max(0,|S|-|n|)$, so that total number of ways to choose $(t,u,v,w) \in S^4$ with $t-v=w-u$ is $\sum_{n \in \Z} \max(0,|S|-|n|)^2$, which is $1^2+2^2+\cdots+(|S|-1)^2+|S|^2+(|S|-1)^2+\cdots+2^2+1^2=(2|S|^3+|S|)/3$.
\end{proof}

\begin{lemma}
Let $S$ be a finite set of consecutive integers.
Let $m$ be a positive integer, let $a\in \Z$, and let $I=\{(t,u,v,w) \in S^4: t+u=v+w, v-t \equiv a \!\! \pmod{m}\}$ and $J=\{(t,u,v,w) \in S^4: t+u=v+w, w-t \equiv a \!\! \pmod{m}\}$.
Then $\card{I}=\card{J}=|S|^2 \Omega\left(\frac{m}{|S|},\frac{a}{|S|}\right)$.
\end{lemma}
\begin{proof}
By exchanging the roles of $v$ and $w$, we see that $\card{I}=\card{J}$, so we simply compute $\card{I}$.
For each $n \in \Z$, the number of pairs $(t,v) \in S^2$ with $v-t=a+n m$ is $\max(0,|S|-|a+n m|)$, so that the number of ways to choose $(t,u,v,w) \in S^4$ with $v-t=u-w=a+n m$ is $\sum_{n \in \Z} \max(0,|S|-|a+n m|)^2$, which is $|S|^2 \Omega(m/|S|,-a/|S|)=|S|^2 \Omega(m/|S|,a/|S|)$.
\end{proof}

\begin{lemma}\label{Leslie}
Let $S$ be a finite set of consecutive integers of the form $S=\{0,1,\ldots,|S|-1\}$.
Let $m$ be a positive integer, let $b\in \Z$, and let $I=\{(t,u,v,x) \in S^4: t+u=v+w, \, t+u \equiv b \!\! \pmod{m}\}$.
Then $\card{I}=|S|^2 \Omega\left(\frac{m}{|S|},1-\frac{b+1}{|S|}\right)$.
\end{lemma}
\begin{proof}
For each $n \in \Z$, the number of pairs $(t,u) \in S^2$ with $t+u=b+n m$ is $\max(0, |S|-|b+n m-(|S|-1)|)$, so that $\card{I}=\sum_{n \in \Z} \max(0,|S|-|b+n m-(|S|-1)|)^2$, which is $|S|^2 \Omega(m/|S|,1-(b+1)/|S|)$.
\end{proof}

\section{Proof of Theorem \ref{Anthony}}\label{Veronica}

In this section, we prove Theorem \ref{Anthony}, whose basic assumptions we recall:
\begin{itemize}
\item We have an infinite family of pairs of sequences $\{(f_\iota,g_\iota)\}_{\iota \in I}$.  Each $f_\iota$ is the multiplicative character sequence associated with character $\phi_\iota$ of prime field $F_\iota$, and $f_\iota$ has shift $r_\iota$ and length $\ell_\iota$.  Each $g_\iota$ is the multiplicative character sequence associated with the character $\chi_\iota$ of field $F_\iota$, and $g_\iota$ has shift $s_\iota$ and length $\ell_\iota$.
\item For every $\iota\in I$, we have the prime $p_\iota=\card{\Fi}$.
\item The set $\{p_\iota\}_{\iota \in I}$ is unbounded, and there is a positive real number $\Lambda$ such that in the limit as $p_\iota \to \infty$, have $\ell_\iota/p_\iota \to \Lambda$.
\end{itemize}
In certain subcases of Theorem \ref{Anthony}, there are additional assumptions, which will be recalled when we address these cases in turn.

Let $(f,g)$ be a pair multiplicative character sequences from our family, where $f$ is associated to the character $\phi$ of the prime field $F$, and $f$ has shift $r$ and length $\ell$, while $g$ is associated to the character $\chi$ of the prime field $F$, and $g$ has shift $s$ and length $\ell$.  Let the prime $p=\card{F}$.
Let $L=\{0,1,\ldots,\ell-1\}$.
So then $f=(f_0,\ldots,f_{\ell-1})$ and $g=(g_0,\ldots,g_{\ell-1})$ with
\begin{align*}
f_j & = \phi(j+r) \\
g_j & = \chi(j+s)
\end{align*}
for every $j \in L$.

We begin by calculating $\normtt{f g}$, the numerator of $\CDF(f,g)$ in \eqref{Melissa}.
By \eqref{Abraham} we have
\[
\normtt{f g}= \sums{t,u,v,w \in L \\ t+u=v+w} \phi(t+r) \chi(u+s) \conj{\phi(v+r) \chi(w+s)}.
\]
We now use Lemma \ref{Theresa} to express the multiplicative characters $\phi$ and $\chi$ in terms of additive characters to see that $\normtt{f g}$ is
\[
\sums{t,u,v,w \in L \\t+u=v+w} \sums{a,b \in F \\ c,d \in F} \frac{G_a(\phi) G_b(\chi) \conj{G_c(\phi) G_d(\chi)}}{p^4} \conj{\epsilon_a(t+r) \epsilon_b(u+s)} \epsilon_c(v+r) \epsilon_d(w+s),
\]
where we recall that $\epsilon$ is the canonical additive character of $F$ and $\epsilon_x(y)=\epsilon(x y)$.
For any $x \in F$, the quadruple $(a+x,b+x,c+x,d+x)$ runs through $F^4$ as $(a,b,c,d)$ runs through $F^4$, so the inner sum is equal to
\[
\sums{a,b,c,d \in F} \frac{G_{a+x}(\phi) G_{b+x}(\chi) \conj{G_{c+x}(\phi) G_{d+x}(\chi)}}{p^4} \conj{\epsilon_a(t+r) \epsilon_b(u+s)} \epsilon_c(v+r) \epsilon_d(w+s)
\]
for every $x \in F$, where we have used the fact that the product of character values $\conj{\epsilon_{a+x}(t+r) \epsilon_{b+x}(u+s)} \epsilon_{c+x}(v+r) \epsilon_{d+x}(w+s)$ is equal to the product of character values $\conj{\epsilon_a(t+r) \epsilon_b(u+s)} \epsilon_c(v+r) \epsilon_d(w+s)$ because $t+u=v+w$ in the outer sum.
So we can average over all $x$ and see that the inner sum is equal to
\[
\sum_{a,b,c,d,x \in F} \!\!\! \frac{G_{a+x}(\phi) G_{b+x}(\chi) \conj{G_{c+x}(\phi) G_{d+x}(\chi)}}{p^5} \conj{\epsilon_a(t+r) \epsilon_b(u+s)} \epsilon_c(v+r) \epsilon_d(w+s)
\]
We therefore define
\[
H(a,b,c,d)=\frac{1}{p^3} \sum_{x\in F} G_{a+x}(\phi) G_{b+x}(\chi) \conj{G_{c+x}(\phi) G_{d+x}(\chi)},
\]
and then
\[
\normtt{f g} = \sums{a,b,c,d \in F} \frac{H(a,b,c,d)}{p^2} \sums{t,u,v,w \in L \\ t+u=v+w} \conj{\epsilon_a(t+r)\epsilon_b(u+s)} \epsilon_c(v+r)\epsilon_d(w+s).
\]
By Lemmata \ref{Gordon} and \ref{Gabriel}\eqref{Katherine} , we can write $H(a,b,c,d)=M(a,b,c,d)+E(a,b,c,d)$, where 
\[
M(a,b,c,d)=\begin{cases}
1 & \text{if $a=c$ and $b=d$,} \\
1 & \text{if $a=d$, $b=c$, and $\phi=\chi$,} \\
1 & \text{if $a=b$, $c=d$, and $\phi=\conj{\chi}$,} \\
0 & \text{otherwise}
\end{cases}
\]
is regarded as the main term and $E(a,b,c,d)$ is regarded as the error term with
\[
|E(a,b,c,d)| \leq \frac{3}{\sqrt{p}}.
\]
So $\normtt{f g} = M_0+E_0$, where
\begin{align*}
M_0 & = \frac{1}{p^2} \sum_{a,b,c,d \in F} M(a,b,c,d) \sums{t,u,v,w \in L \\ t+u=v+w} \conj{\epsilon_a(t+r)\epsilon_b(u+s)} \epsilon_c(v+r)\epsilon_d(w+s) \\
E_0 & = \frac{1}{p^2} \sum_{a,b,c,d \in F} E(a,b,c,d) \sums{t,u,v,w \in L \\ t+u=v+w} \conj{\epsilon_a(t+r)\epsilon_b(u+s)} \epsilon_c(v+r)\epsilon_d(w+s),
\end{align*}
and our bound on $E(a,b,c,d)$ tells us that
\[
|E_0| \leq \frac{3}{p^{5/2}} \sums{a,b,c,d \in F} \left|\sums{t,u,v,w \in L \\ t+u=v+w} \conj{\epsilon_a(t+r)\epsilon_b(u+s)} \epsilon_c(v+r)\epsilon_d(w+s)\right|,
\]
and so
\begin{equation}\label{Eric}
|E_0| \leq 192 p^{3/2} (1+\log p)^3 \max\left(1,\frac{\ell}{p}\right)^3
\end{equation}
by Lemma \ref{William} below.

The main term $M_0$ can be broken into four terms $M_1$, $M_2$, $M_3$, and $M_4$ based on the four mutually exclusive cases where $M(a,b,c,d)$ might be nonzero (1) $a=b=c=d$, (2) $a=c\not=b=d$, (3) $a=d\not=b=c$, and (4) $a=b\not=c=d$.
Then
\[
\normtt{f g} = M_0+E_0 = M_1+M_2+M_3+M_4+E_0,
\]
where
\begin{align*}
M_1 & = \frac{1}{p^2} \sum_{a\in F} \sums{t,u,v,w \in L \\ t+u=v+w} 1, \\
M_2 & = \frac{1}{p^2} \sums{a,b \in F \\ a\not=b} \sums{t,u,v,w \in L \\ t+u=v+w} \epsilon_a(v-t) \epsilon_b(w-u) \\
M_3 & = \frac{\delta_{\phi,\chi}}{p^2} \sums{a,b \in F \\ a\not=b} \sums{t,u,v,w\in L \\ t+u=v+w} \epsilon_a(w-t+s-r) \epsilon_b(v-u+r-s) \\
M_4 & = \frac{\delta_{\phi,\conj{\chi}}}{p^2} \sums{a,c \in F \\ a\not=c} \sums{t,u,v,w \in L \\ t+u=v+w} \conj{\epsilon_a(t+u+r+s)} \epsilon_c(v+w+r+s),
\end{align*}
and where $\delta$ is the Kronecker delta.
Let
\begin{align*}
A & = \{(t,u,v,w) \in L^4: t+u=v+w\} \\
B & = \{(t,u,v,w) \in A: v-t \equiv 0 \!\! \pmod{p}\} \\
C & = \{(t,u,v,w) \in A: w-t \equiv r-s \!\! \pmod{p}\} \\
D & = \{(t,u,v,w) \in A: t+u \equiv -(r+s) \!\! \pmod{p}\},
\end{align*}
so that, by the orthogonality relation in Lemma \ref{Orestes}\eqref{Sarah}, we have
\begin{align*}
M_1 & = \frac{\card{A}}{p} \\
M_2 & = \card{B} -\frac{\card{A}}{p} \\
M_3 & = \delta_{\phi,\chi} \left(\card{C}-\frac{\card{A}}{p}\right)\\
M_4 & = \delta_{\phi,\conj{\chi}} \left(\card{D} - \frac{\card{A}}{p}\right),
\end{align*}
and so $M_0=M_1+M_2+M_3+M_4=M_A+M_B+M_C+M_D$, where
\begin{align*}
M_A & =  -\frac{(\delta_{\phi,\chi}+\delta_{\phi,\conj{\chi}}) \card{A}}{p} \\
M_B & = \card{B} \\
M_C & = \delta_{\phi,\chi} \card{C} \\
M_D & = \delta_{\phi,\conj{\chi}} \card{D}.
\end{align*}
Recall the function $\Omega$ defined in \eqref{Orlando}.
Using the cardinalities of $A$, $B$, $C$, and $D$ computed in Lemmata \ref{Iago}--\ref{Leslie}, we obtain $\normtt{f g}=M_0+E_0=M_A+M_B+M_C+M_D+E_0$ with
\begin{align*}
M_A & = -\frac{(\delta_{\phi,\chi}+\delta_{\phi,\conj{\chi}})(2\ell^3+\ell)}{3 p} \\
M_B & = \ell^2 \Omega\left(\frac{p}{\ell},0\right) \\
M_C & = \delta_{\phi,\chi} \ell^2 \Omega\left(\frac{p}{\ell},\frac{r-s}{\ell}\right) \\
M_D & = \delta_{\phi,\conj{\chi}} \ell^2 \Omega\left(\frac{p}{\ell},1+\frac{r+s-1}{\ell}\right),
\end{align*}
and where $E_0$ is bounded in \eqref{Eric}.
Thus we have
\[
\frac{\normtt{f g}}{\ell^2}=N_A+N_B+N_C+N_D+E_1
\]
with
\begin{align*}
N_A & = -\frac{(\delta_{\phi,\chi}+\delta_{\phi,\conj{\chi}})(2\ell^2+1)}{3 p\ell} \\
N_B & = \Omega\left(\frac{p}{\ell},0\right) \\
N_C & = \delta_{\phi,\chi} \Omega\left(\frac{p}{\ell},\frac{r-s}{\ell}\right) \\
N_D & = \delta_{\phi,\conj{\chi}} \Omega\left(\frac{p}{\ell},1+\frac{r+s-1}{\ell}\right) \\
|E_1| & \leq \frac{192 p^{3/2} (1+\log p)^3}{\ell^2} \max\left(1,\frac{\ell}{p}\right)^3.
\end{align*}
Now that we have estimated $\normtt{f g}/\ell^2$ for a single pair $(f,g)$ from the family $\{(f_\iota,g_\iota)\}_{\iota \in I}$, let us consider the asymptotic behavior of $\normtt{f g}/\ell^2$ as $p \to \infty$.
By Lemma \ref{Martin} below we see that $\normtt{f}/\ell$ and $\normtt{g}/\ell$ both tend to $1$ as $p \to\infty$, so the asymptotic behavior of $\cnr{f}{g}$ is the same as that of $\normtt{f g}/\ell^2$.

The ratio $\ell/p$ tends to the positive real number $\Lambda$, so if all our sequences $f_\iota$ and $g_\iota$ are associated to the quadratic character, then $\phi_\iota=\chi_\iota=\conj{\chi_\iota}$ for all $\iota \in I$, so then $N_A \to -4\Lambda/3$ as $p\to \infty$.  If all $f_\iota$ and $g_\iota$ are always associated to the same nonquadratic character, then $\phi_\iota=\chi_\iota\not=\conj{\chi_\iota}$, and then $N_A \to -2\Lambda/3$ as $p \to \infty$.  If all $f_\iota$ and $g_\iota$ are always associated to conjugate nonquadratic characters, then $\phi_\iota=\conj{\chi_\iota}\not=\chi_\iota$, so $N_A \to -2\Lambda/3$ as $p \to \infty$.   If $f_\iota$ and $g_\iota$ are always associated to characters that are neither equal nor conjugate, then $N_A$ is always $0$.

Recall that $\Omega$ is continuous on $\{(x,y): x\not=0\}$, so that $N_B \to \Omega(1/\Lambda,0)$ as $p\to \infty$.

If $f_\iota$ and $g_\iota$ are always associated to the same (quadratic or nonquadratic) character, then we are given the extra assumption that $(r-s)/p \to \Delta$ as $p\to\infty$, and then $N_C \to \Omega(1/\Lambda,\Delta/\Lambda)$ as $p\to\infty$.  If $f_\iota$ and $g_\iota$ always have different characters, then $N_C$ is always $0$.

If $f_\iota$ and $g_\iota$ are always associated to conjugate (quadratic or nonquadratic) characters, then we are given the extra assumption that $(r+s)/p \to \Sigma$ as $p\to\infty$, and then $N_D \to \Omega(1/\Lambda,1+\Sigma/\Lambda)$ as $p\to\infty$.  If $f_\iota$ and $g_\iota$ always have nonconjugate characters, then $N_D$ is always $0$.

The term $E_1 \to 0$ as $p\to\infty$ because $\ell/p$ tends to a positive real number $\Lambda$ as $p \to\infty$.

This completes the proofs of the limiting values of $\cnr{f_\iota}{g_\iota}$ as $p_\iota \to \infty$ when the sequences $f_\iota$ and $g_\iota$ are multiplicative character sequences.  If we replace these sequences with unimodularizations, then we invoke Lemma \ref{Oswald} to see that the limits do not change.  This completes the proof of Theorem \ref{Anthony} \hfill\qedsymbol

We conclude this section with technical lemmata that we used in our proof of Theorem \ref{Anthony}.
\begin{lemma}\label{Gordon}
Let $F$ be a finite field, let $a,b,c,d \in F$, let $\phi,\chi,\psi,\omega$ be nontrivial multiplicative characters of $F$, and let
\[
S=\frac{1}{\card{F}^3} \sum_{x \in F} G_{a+x}(\phi) G_{b+x}(\chi) \conj{G_{c+x}(\psi) G_{d+x}(\omega)}.
\]
Let
\[
M=\begin{cases}
\frac{G(\phi)G(\chi)\conj{G(\psi)G(\omega)}}{\card{F}^2}  & \text{if $a=b=c=d$ and $\phi\chi=\psi\omega$,} \\
\phi\psi(-1)  & \text{if $a=b\not=c=d$, $\phi=\conj{\chi}$, and $\psi=\conj{\omega}$}, \\
1 & \text{if $a=c\not=b=d$, $\phi=\psi$, and $\chi=\omega$}, \\
1 & \text{if $a=d\not=b=c$, $\phi=\omega$, and $\chi=\psi$}, \\
0 & \text{otherwise}.
\end{cases}
\]
Then $S=M+E$ for some $E \in \C$ with $|E| \leq 3/\sqrt{\card{F}}$.
\end{lemma}
\begin{proof}
We use Lemma \ref{Gabriel}\eqref{Lawrence} to see that
\begin{equation}\label{Emily}
S=\frac{G(\phi)G(\chi)\conj{G(\psi) G(\omega)}}{\card{F}^3} \sum_{x \in F} \conj{\phi(a+x) \chi(b+x)} \psi(c+x) \omega(d+x).
\end{equation}
Now we apply the Weil bound \cite{Weil} to the sum over $x \in F$.  Since all our characters are nontrivial, this sum will be at most $3\sqrt{\card{F}}$ unless $a,b,c,d$ can be arranged into two pairs of equal values.
This bound and the bound of Lemma \ref{Gabriel}\eqref{Katherine} give the desired result when $a,b,c,d$ cannot be paired thus.

When $a=b=c=d$, the sum over $x \in F$ in \eqref{Emily} vanishes unless $\phi\chi=\psi\omega$.
So if $\phi\chi\not=\psi\omega$, we again obtain the desired result.
If $\phi\chi=\psi\omega$, then the sum over $x \in F$ will be $\card{F}-1$ by Lemma \ref{Orestes}\eqref{Theodore}, and we obtain
\[
S=\frac{\card{F}-1}{\card{F}^3} G(\phi)G(\chi)\conj{G(\psi)G(\omega)},
\]
which differs from $G(\phi)G(\chi)\conj{G(\psi)G(\omega)}/\card{F}^2$ by a term of magnitude $1/\card{F}$ by Lemma \ref{Gabriel}\eqref{Katherine}.

When $a=b\not=c=d$, then the Weil bound says that the sum over $x\in F$ in \eqref{Emily} will be at most $\sqrt{\card{F}}$ unless $\phi=\conj{\chi}$ and $\psi=\conj{\omega}$.
So if $\phi\not=\conj{\chi}$ or $\psi\not=\conj{\omega}$, we obtain the desired result by using Lemma \ref{Gabriel}\eqref{Katherine}.
If $\phi=\conj{\chi}$ and $\psi=\conj{\omega}$, then the sum over $x \in F$ will be $\card{F}-2$, and we obtain
\begin{align*}
S
& =\frac{\card{F}-2}{\card{F}^3} G(\phi)G(\conj{\phi})\conj{G(\psi)G(\conj{\psi})} \\
& =\frac{\card{F}-2}{\card{F}} \phi(-1) \psi(-1),
\end{align*}
where we have used Lemma \ref{Gabriel}\eqref{Manuel},\eqref{Katherine} in the second equality.
This differs from $\phi\psi(-1)$ by a term of magnitude $2/\card{F}$.

When $a=c\not=b=d$, then the Weil bound says that the sum over $x \in F$ in \eqref{Emily} will be at most $\sqrt{\card{F}}$ unless $\phi=\psi$ and $\chi=\omega$.
So if $\phi\not=\psi$ or $\chi\not=\omega$, we again obtain the desired result by using Lemma \ref{Gabriel}\eqref{Katherine}.
If $\phi=\psi$ and $\chi=\omega$, then the sum over $x \in F$ will be $\card{F}-2$, and we obtain
\begin{align*}
S
& =\frac{\card{F}-2}{\card{F}^3} G(\phi)G(\chi)\conj{G(\phi)G(\chi)} \\
& =\frac{\card{F}-2}{\card{F}} 
\end{align*}
where we have used Lemma \ref{Gabriel}\eqref{Katherine} in the second equality.
This differs from $1$ by a term of magnitude $2/\card{F}$.

When $a=d\not=b=c$, then the Weil bound says that the sum over $x \in F$ in \eqref{Emily} will be at most $\sqrt{\card{F}}$ unless $\phi=\omega$ and $\chi=\psi$, and the analysis of this case is the same as the previous one, except that $\psi$ and $\omega$ exchange roles.
\end{proof}

\begin{lemma}\label{William}
Let $F$ be the prime field of order $p$, let $\epsilon$ be the canonical additive character of $F$, and let $\epsilon_u$ be the additive character $x \mapsto \epsilon(u x)$.
Let $h,i,j,k,\ell \in \Z$ and let $L=\{0,1,\ldots,\ell-1\}$.  Then the sum
\[
S= \sum_{a,b,c,d \in F} \left|\sums{t,u,v,w \in L \\ t+u=v+w} \conj{\epsilon_a(t+h)\epsilon_b(u+i)} \epsilon_c(v+j)\epsilon_d(w+k)\right|
\]
satisfies $|S| \leq 64 p^4 (1+\log p)^3 \max(1,\ell/p)^3$.
\end{lemma}
\begin{proof}
Since $\conj{\epsilon_a(h)}$, $\conj{\epsilon_b(i)}$, $\epsilon_c(j)$, and $\epsilon_d(k)$ are always of unit magnitude, we have
\[
S=\sum_{a,b,c,d \in F} \left|\sums{t,u,v,w \in L \\ t+u=v+w} \conj{\epsilon_a(t) \epsilon_b(u)} \epsilon_c(v) \epsilon_d(w)\right|,
\]
and then note that $\epsilon_a$ runs through $\achars$ as $a$ runs through $F$, and similarly with $b$, $c$, and $d$.
Then we employ \cite[Lemma A.1]{Katz-2013} (using $\Gamma=F$, $n=1$, with the map $\pi_1\colon \Z \to F$ being reduction modulo $p$, and $U=L=\{0,1,\ldots,\ell-1\}$) to see that the last expression for $S$ is bounded above by $64 p^4 (1+\log p)^3 \max(1,\ell/p)^3$.
\end{proof}

\begin{lemma}\label{Martin}
Let $\{f_\iota\}_{\iota\in I}$ be an infinite family of sequences, where for each $\iota \in I$, there is a prime field $\Fi$ of order $p_\iota$ and a positive integer $\ell_\iota$ such that $f_\iota$ is a multiplicative character sequence of length $\ell_\iota$ associated to some character of $\Fi$.
Suppose that $\{p_\iota\}_{\iota \in I}$ is unbounded and that there is a positive real number $\Lambda$ such that $\ell_\iota/p_\iota$ tends to $\Lambda$ as $p_\iota \to \infty$.
Then there is a $P$ and an $N$ such that for all $\iota \in I$ with $p_\iota \geq P$, not more than $N$ terms of the sequence $f_\iota$ are equal to $0$.
Then $\normtt{f_\iota}/\ell_\iota \to 1$ as $p_\iota \to \infty$.
If $\tilde{f}_\iota$ is a unimodularization of $f_\iota$ for each $\iota \in I$, then $\normt{f_\iota}/\normt{\tilde{f}_\iota} \to 1$ as $p_\iota \to \infty$.
\end{lemma}
\begin{proof}
This is a special case of \cite[Lemma A.2, Corollary A.3]{Katz-2013}, where we let $e=1$ and set $q_\iota=p_\iota$, $S_\iota=\{0,1,\ldots,\ell_\iota-1\}$, and let $\alpha_\iota\colon \Z\to\Fi$ be the map that effects reduction modulo $p_\iota$.
\end{proof}

\begin{lemma}\label{Oswald}
Let $\{(f_\iota,g_\iota)\}_{\iota \in I}$ be an infinite family of pairs of sequences, where for each $\iota \in I$, there is a prime field $\Fi$ of order $p_\iota$ and a positive integer $\ell_\iota$ such that $f_\iota$ and $g_\iota$ are multiplicative character sequences of length $\ell_\iota$ associated to some characters of $\Fi$.
Let $\tilde{f}_\iota$ and $\tilde{g}_\iota$ be unimodularizations of $f_\iota$ and $g_\iota$, respectively, for each $\iota \in I$.
Suppose that $\{p_\iota\}_{\iota \in I}$ is unbounded and that there is a positive real number $\Lambda$ such that $\ell_\iota/p_\iota$ tends to $\Lambda$ as $p_\iota \to \infty$.
Then $\CDF(f_\iota,g_\iota)$ tends to a real limit as $p_\iota \to \infty$ if and only if $\CDF(\tilde{f}_\iota,\tilde{g}_\iota)$ tends to a real limit as $p_\iota \to \infty$, and in this case, the limits are equal.
\end{lemma}
\begin{proof}
Note that if $u(z)=\sum_{j=0}^d u_j z^j$ and $v(z)$ are polynomials with complex coefficients, then the triangle inequality for $L^2$ norms gives
\begin{align}
\normt{u(z) v(z)}
& \leq \sum_{j=0}^d \normt{u_j z^j v(z)}  \label{Horace} \\
& = \normt{v(z)} \sum_{j=0}^{d} |u_j| \nonumber
\end{align}
where the last step uses the fact that $|u_j z^j|=|u_j|$ for $z$ on the complex unit circle.

Lemma \ref{Martin} furnishes a $P$ and an $N$ such that for every $\iota \in I$ with $p_\iota \geq P$, each of the sequences $f_\iota$ and $g_\iota$ has at most $N$ of its terms equal to zero.
Thus $\tilde{f}_\iota$ and $f_\iota$ differ in at most $N$ terms and $\tilde{g}_\iota$ and $g_\iota$ differ in at most $N$ terms, where the unimodularizations have a unimodular complex number in place of a zero. 
We consider only $\iota \in I$ with $p_\iota \geq P$ for the rest of this proof.

By the triangle inequality for $L^2$ norms, 
\begin{align*}
\Big|\normt{\tilde{f}_\iota \tilde{g}_\iota}-\normt{f_\iota g_\iota} \Big|
& \leq \normt{(\tilde{f}_\iota-f_\iota)\tilde{g}_\iota} + \normt{f_\iota (\tilde{g}_\iota-g_\iota)} \\
& \leq N \normt{\tilde{g}_\iota} + N \normt{f_\iota} \\
& \leq 2 N \sqrt{\ell_\iota},
\end{align*}
where the second inequality uses \eqref{Horace} and the fact that (when regarded as polynomials) each of $\tilde{f}_\iota-f_\iota$ and $\tilde{g}_\iota - g_\iota$ is a sum of at most $N$ monomials with unimodular coefficients, and the third equality uses the fact that the squared $L^2$ norm of a polynomial is the sum of the squared magnitudes of its coefficients, and each of $\tilde{g}_\iota$ and $f_\iota$ is the sum $\ell_\iota$ monomials with unimodular or zero coefficients.

Then
\[
\frac{\Big|\normt{\tilde{f}_\iota \tilde{g}_\iota}-\normt{f_\iota g_\iota} \Big|}{\normt{f_\iota} \normt{g_\iota}} \leq \frac{2 N \sqrt{\ell_\iota}}{\normt{f_\iota}\normt{g_\iota}},
\]
and the right hand side tends to $0$ as $p_\iota \to \infty$ because Lemma \ref{Martin} shows that both $\sqrt{\ell_\iota}/\normt{f_\iota}$ and $\sqrt{\ell_\iota}/\normt{g_\iota}$ tend to $1$ as $p_\iota \to \infty$, and yet $\ell_\iota \to \infty$ as $p_\iota \to \infty$ (because $\ell_\iota/p_\iota$ tends to the positive number $\Lambda$ in this limit).
So $\normt{\tilde{f}_\iota \tilde{g}_\iota}/(\normt{f_\iota} \normt{g_\iota})$ tends to a real limit as $p_\iota \to \infty$ if and only if $\normt{f_\iota g_\iota}/(\normt{f_\iota} \normt{g_\iota})$ tends to a real limit as $p_\iota \to \infty$, and if these limits exist, then they are equal.
By Lemma \ref{Martin}, we know that $\normt{f_\iota}/\normt{\tilde{f}_\iota}$ and $\normt{g_\iota}/\normt{\tilde{g}_\iota}$ both tend to $1$ as $p_\iota \to \infty$, so we may conclude that $\normt{\tilde{f}_\iota \tilde{g}_\iota}/(\normt{\tilde{f}_\iota} \normt{\tilde{g}_\iota})$ and $\normt{f_\iota g_\iota}/(\normt{f_\iota} \normt{g_\iota})$ have the same limiting behavior as $p_\iota \to \infty$.
And by squaring both we see that $\CDF(\tilde{f}_\iota,\tilde{g}_\iota)$ and $\CDF(f_\iota,g_\iota)$ have the same limiting behavior as $p_\iota \to \infty$.
\end{proof}
\section{Optimization of Crosscorrelation Demerit Factors}\label{Elizabeth}

In this section we find the global minimum crosscorrelation demerit factors (or equivalently, global maximum crosscorrelation merit factors) for the formulae in Theorems \ref{Victor} and \ref{Anthony}.
Throughout this section $\Rplus$ denotes the set of positive real numbers.
Recall the definition \eqref{Orlando} of the function $\Omega(x,y)$.

The following gives the global minimum crosscorrelation demerit factor for Theorem \ref{Victor}\eqref{Alan} and for Theorem \ref{Anthony}\eqref{Bernice}.
\begin{lemma}\label{Jake}
The function $f\colon \Rplus \times \R \to \R$ with 
\[
f(\Lambda,\Delta)=-\frac{2}{3} \Lambda + \Omega\left(\frac{1}{\Lambda},0\right) + \Omega\left(\frac{1}{\Lambda},\frac{\Delta}{\Lambda}\right).
\]
achieves a global minimum value of $0.649608\ldots$, the smallest root of the polynomial $6 x^3-42 x +54 x-19$.
This global minimum is attained if and only if $(\Lambda,\Delta) \in \{(\Lambdatrunc,m+1/2): m \in \Z\}$, where $\Lambdatrunc=0.557874\ldots$, the middle root of the polynomial $2 x^3-6 x+3$.  Equivalently the function $1/f(\Lambda,\Delta)$ achieves a global maximum value of $1.539389\ldots$, the largest root of the polynomial $19 x^3-54 x^2+42 x-6$ for the same values of $(\Lambda,\Delta)$.
\end{lemma}
\begin{proof}
For any given value of $\Lambda$, we can use Lemma \ref{Rachel} below to find the values of $\Delta$ that globally minimize $\Omega(1/\Lambda,\Delta/\Lambda)$, namely $\Delta\in \{m+1/2: m \in \Z\}$.
This reduces our problem to the minimization of the single-variable function $g\colon \Rplus \to \R$ with
\[
g(\Lambda)=-\frac{2}{3} \Lambda + \Omega\left(\frac{1}{\Lambda},0\right) + \Omega\left(\frac{1}{\Lambda},\frac{1}{2\Lambda}\right),
\]
which achieves the claimed global minimum value when $\Lambda=\Lambdatrunc$ by Lemma \ref{Olivia}.
\end{proof}
The following gives the global minimum crosscorrelation demerit factor for Theorem \ref{Victor}\eqref{Reginald} and for Theorem \ref{Anthony}\eqref{Nora}.
\begin{lemma}\label{Barbara}
The function $f\colon \Rplus \times \R \to \R$ with 
\[
f(\Lambda,\Sigma)=-\frac{2}{3} \Lambda + \Omega\left(\frac{1}{\Lambda},0\right) + \Omega\left(\frac{1}{\Lambda},1+\frac{\Sigma}{\Lambda}\right).
\]
achieves a global minimum value of $0.649608\ldots$, the smallest root of the polynomial $6 x^3-42 x +54 x-19$.
This global minimum is attained if and only if $(\Lambda,\Sigma) \in \{(\Lambdatrunc,n+1/2-\Lambdatrunc): n \in \Z\}$, where $\Lambdatrunc=0.557874\ldots$, the middle root of the polynomial $2 x^3-6 x+3$.  Equivalently the function $1/f(\Lambda,\Sigma)$ achieves a global maximum value of $1.539389\ldots$, the largest root of the polynomial $19 x^3-54 x^2+42 x-6$ for the same values of $(\Lambda,\Sigma)$.
\end{lemma}
\begin{proof}
For any given value of $\Lambda$, we can use Lemma \ref{Rachel} below to find the values of $\Sigma$ that globally minimize $\Omega(1/\Lambda,1+\Sigma/\Lambda)$, namely $\Sigma\in \{n+1/2-\Lambda: n \in \Z\}$.
This reduces our problem to the minimization of the single-variable function $g\colon \Rplus \to \R$ with
\[
g(\Lambda)=-\frac{2}{3} \Lambda + \Omega\left(\frac{1}{\Lambda},0\right) + \Omega\left(\frac{1}{\Lambda},\frac{1}{2\Lambda}\right),
\]
which achieves the claimed global minimum value when $\Lambda=\Lambdatrunc$ by Lemma \ref{Olivia}.
\end{proof}
The following gives the global minimum crosscorrelation demerit factor for Theorem \ref{Victor}\eqref{Ophelia} and for Theorem \ref{Anthony}\eqref{Walter}.
\begin{lemma}\label{Leonard}
The function $f\colon \Rplus \to \R$ with
\[
f(\Lambda)= \Omega\left(\frac{1}{\Lambda},0\right)
\]
achieves a global minimum value of $1$.  This global minimum is attained if and only if $\Lambda \in (0,1]$.  Equivalently the function $1/f(\Lambda)$ achieves a global maximum value of $1$ for the same values of $\Lambda$.
\end{lemma}
\begin{proof}
This is immediate from the calculations in Lemma \ref{Mary}.
\end{proof}
The following gives the global minimum crosscorrelation demerit factor for Theorem \ref{Anthony}\eqref{Hubert}.
\begin{lemma}\label{Christopher}
The function $f\colon \Rplus \times \R \times \R \to \R$ with 
\[
f(\Lambda,\Delta,\Sigma)=-\frac{4}{3} \Lambda + \Omega\left(\frac{1}{\Lambda},0\right) + \Omega\left(\frac{1}{\Lambda},\frac{\Delta}{\Lambda}\right) + \Omega\left(\frac{1}{\Lambda},1+\frac{\Sigma}{\Lambda}\right).
\]
achieves a global minimum value of $0.299216\ldots$, the smallest root of the polynomial $3 x^3-33 x +33 x-7$.  This global minimum is attained if and only if $(\Lambda,\Delta,\Sigma) \in \{(\Lambdatrunc,m+1/2,n+1/2-\Lambdatrunc): m,n \in \Z\}$, where $\Lambdatrunc=0.557874\ldots$, the middle root of the polynomial $2 x^3-6 x+3$.  Equivalently the function $1/f(\Lambda,\Delta,\Sigma)$ achieves a global maximum value of $3.342065\ldots$, the largest root of the polynomial $7 x^3-33 x^2+33 x-3$ for the same values of $(\Lambda,\Delta,\Sigma)$.
\end{lemma}
\begin{proof}
For any given value of $\Lambda$, we can use Lemma \ref{Rachel} below to find the values of $\Delta$ and $\Sigma$ that globally minimize $\Omega(1/\Lambda,\Delta/\Lambda)$ and $\Omega(1/\Lambda,1+\Sigma/\Lambda)$, respectively, namely $\Delta\in \{m+1/2: m \in \Z\}$ and $\Sigma \in \{n+1/2-\Lambda: n \in \Z\}$.
This reduces our problem to the minimization of the single-variable function $g\colon \Rplus \to \R$ with
\[
g(\Lambda)=-\frac{4}{3} \Lambda + \Omega\left(\frac{1}{\Lambda},0\right) + 2 \Omega\left(\frac{1}{\Lambda},\frac{1}{2\Lambda}\right),
\]
which achieves the claimed global minimum value when $\Lambda=\Lambdatrunc$ by Lemma \ref{Olivia}.
\end{proof}

We conclude this section with the technical lemmata used to prove Lemmata \ref{Jake}--\ref{Christopher}.
\begin{lemma}\label{Rachel}
Let $x$ be a fixed nonzero real number, and let $f \colon \R \to \R$ with $f(y)=\Omega(x,y)$.
If $0 < |x| \leq 2$, then the global minimum value of $f(y)$ is $\Omega(x,x/2)$ and is attained if and only if $y \in \{(n+1/2)|x|: n \in \Z\}$.
If $|x| \geq 2$, then the global minimum value of $f(y)$ is $0$ (which equals $\Omega(x,x/2)$), and is attained if and only if $y \in \bigcup_{n \in \Z} [n|x|+1,(n+1)|x|-1]$.
\end{lemma}
\begin{proof}
Everything stated here is proved in \cite[Lemma A.4]{Katz-2013}, except for the trivial observation that $\Omega(x,x/2)=0$ when $|x| \geq 2$.
\end{proof}

\begin{lemma}\label{Olivia}
Consider the functions $f\colon \Rplus \to \R$ and $g\colon \Rplus \to \R$ with
\begin{align*}
f(x) & =-\frac{4}{3} x+ \Omega\left(\frac{1}{x},0\right) + 2 \Omega\left(\frac{1}{x},\frac{1}{2 x}\right) \\
g(x) & =-\frac{2}{3} x+ \Omega\left(\frac{1}{x},0\right) + \Omega\left(\frac{1}{x},\frac{1}{2 x}\right).
\end{align*}
Each of these functions achieves its unique global minimum value when $x=0.557874\ldots$, the middle root of the polynomial $2 X^3-6 X+3$.  The global minimum value of $f$ is $0.299216\ldots$, the smallest root of the polynomial $3 X^3-33 X +33 X-7$, while the global minimum value of $g$ is $0.649608\ldots$, the smallest root of the polynomial $6 X^3-42 X^2+54 X-19$.
\end{lemma}
\begin{proof}
Lemmata \ref{Mary} and \ref{Nancy} tell us that $f^\prime(x)=-4/3$ when $x \in (0,1/2]$.

If $m$ is a positive integer and $x \in [m,m+1/2]$, then Lemmata \ref{Mary} and \ref{Nancy} tell us that
\begin{align*}
f'(x) & = -\frac{4}{3} + \frac{2 m(m+1)}{x^2} - \frac{2 m(m+1)(2 m+1)}{3 x^3} + \frac{4 m^2}{x^2} - \frac{2 m(4 m^2-1)}{3 x^3} \\
& = -\frac{4}{3} + \frac{6 m^2+2 m}{x^2} - \frac{12 m^3+6 m^2}{3 x^3},
\end{align*}
which has the same sign as $h(x)=3 x^3 f^\prime(x)=-4 x^3 + (18 m^2+6 m) x -(12 m^3+6 m^2)$.  Note that $h(m)=2 m^3$ and $h'(x)=-12 x^2+ 18 m^2+6 m$.  So if $m \geq 2$, we see that $h(x)> 0$ for all $x \in [m,m+1/2]$ because $h(m) > 0$ and $h'(x) \geq h'(m+1/2) = 6 m(m-1)-3 > 0$ for all $x \in [m,m+1/2]$.
On the other hand, if $m=1$, then $h(x)=-4 x^3+24 x-18$ and $h'(x)=-12 x^2+24$, so that $h(x)$ is increasing on $[1,\sqrt{2})$ and then decreasing on $(\sqrt{2},1]$ and so for $x \in [1,3/2]$ we have $h(x) \geq \min(h(1),h(3/2))=\min(2,9/2)=2$ so $h(x) > 0$ for all $x \in [1,3/2]$.
So $h(x)> 0$ (and hence $f^\prime(x) > 0$) for all $x \in [m,m+1/2]$ for every positive integer $m$.

If $m$ is a positive integer and $x \in [m-1/2,m]$, then Lemmata \ref{Mary} and \ref{Nancy} tell us that
\begin{align*}
f'(x) & = -\frac{4}{3} + \frac{2 (m-1)m}{x^2} - \frac{2 (m-1)m(2 m-1)}{3 x^3} + \frac{4 m^2}{x^2} - \frac{2 m(4 m^2-1)}{3 x^3} \\
& = -\frac{4}{3} + \frac{6 m^2-2 m}{x^2} - \frac{12 m^3-6 m^2}{3 x^3},
\end{align*}
which has the same sign as $h(x)=3 x^3 f^\prime(x)=-4 x^3 + (18 m^2-6 m) x -(12 m^3-6 m^2)$.  Note that $h(m-1/2)=2 m^3-3 m^2+1/2$ and $h'(x)=-12 x^2+ 18 m^2-6 m$.  So if $m \geq 2$, we see that $h(x)> 0$ for all $x \in [m-1/2,m]$ because $h(m-1/2) > 0$ and $h'(x) \geq h'(m) = 6 m(m-1) >  0$ for all $x \in [m-1/2,m]$.  So $h(x) > 0$ (and hence $f^\prime(x) > 0$) for all $x \in [m-1/2,m]$ when $m \geq 2$.  On the other hand, if $m=1$, then $h(x)=-2(2 x^3-6x+3)$.  Note that the polynomial $P(x)=-2(2 x^3-6 x+3)$ has $P(-2)=2$, $P(1/2)=-1/2$, $P(1)=2$, and $P(2)=-14$, so $P(x)$ has one root in each of the intervals $(-2,1/2)$, $(1/2,1)$, and $(1,2)$ and $\Lambdatrunc$ as described in the statement of this lemma is the middle root.  And so we see that $h(x) < 0$ for $x \in [1/2,\Lambdatrunc)$, $h(\Lambdatrunc)=0$, and $h(x) > 0$ for $x \in (\Lambdatrunc,1]$, and recall that  $f^\prime(x)=h(x)/3 x^3$ always has the same sign as $h(x)$.

Putting all our findings together, we have shown that $f^\prime(x) < 0$ for $x \in (0,\Lambdatrunc)$, $f^\prime(\Lambdatrunc)=0$, and $f^\prime(x) > 0$ for $x \in (\Lambdatrunc,\infty)$, so that $f(x)$ has a unique global minimum at $x=\Lambdatrunc$.  Since $\Lambdatrunc \in (1/2,1)$, Lemmata \ref{Mary} and \ref{Nancy} show that
\[
f(\Lambdatrunc)=-\frac{4}{3} \Lambdatrunc + 1 + 4 - \frac{4}{\Lambdatrunc} + \frac{1}{\Lambdatrunc^2}.
\]
Since $\Lambdatrunc$ satisfies $2 X^3-6 X+3$, we see that $1/\Lambdatrunc=2-2 \Lambdatrunc^2/3$, from which one can deduce (using the relation $2\Lambdatrunc^3-6\Lambdatrunc+3=0$) that $f(\Lambdatrunc)=4\Lambdatrunc^2/3-2\Lambdatrunc+1$, which is a root of the polynomial $3 X^3-33 X^2+33 X-7$.
This polynomial has one root in each of the intervals $(0,2/3)$, $(2/3,1)$, and $(1,10)$.  We see that $f(\Lambdatrunc)$ must be the root in $(0,2/3)$ since $f(\Lambdatrunc) < f(1)=2/3$.

Note that $g(x)=(f(x)+\Omega(1/x,0))/2$, and since Lemma \ref{Mary} tells us that $\Omega(1/x,0)$ always has a nonnegative derivative with respect to $x$, we see that $g^\prime(x)$ is positive whenever $f^\prime(x)$ is positive.  Thus we know that $g^\prime(x) > 0$ for all $x \in (\Lambdatrunc,\infty)$.  Furthermore, Lemma \ref{Mary} tells us that the derivative of $\Omega(1/x,0)$ vanishes for $x \leq 1$, so $g^\prime(x)=f^\prime(x)/2$ for $x \in (0,\Lambdatrunc]$, so that $g^\prime(x) < 0$ for $x \in (0,\Lambdatrunc)$ and $g^\prime(\Lambdatrunc)=0$.  Thus we conclude that $g(x)$ achieves a unique global minimum at $x=\Lambdatrunc$.  And $g(\Lambdatrunc)=(f(\Lambdatrunc)+1)/2$, so $g(\Lambdatrunc)$ is the smallest root of the polynomial $3(2 X-1)^3-33(2 X-1)^2+33(2 X-1)-7=4(6 X^3-42 X^2+54 X-19)$.
\end{proof}

\begin{lemma}\label{Mary}
Let $\Phi\colon \Rplus \to \R$ with $\Phi(x)=\Omega(1/x,0)$.  If $m$ is a nonnegative integer, and $x \in [m,m+1]$, then
\begin{align*}
\Phi(x) & = 2 m+1 - \frac{2 m(m+1)}{x} + \frac{m(m+1)(2 m+1)}{3 x^2} \\
\Phi^\prime(x) & = \frac{2 m(m+1)}{x^2} - \frac{2 m(m+1)(2 m+1)}{3 x^3}.
\end{align*}
We have $\Phi^\prime(x)=0$ for $x \in (0,1]$, while $\Phi^\prime(x) > 0$ for $x > 1$, so that $\Phi(x)=1$ for $x \in (0,1]$ and $\Phi(x) > 1$ for $x > 1$.
\end{lemma}
\begin{proof}
Let $m$ be a nonnegative integer and suppose $x \in [m,m+1]$.  Then the definition \eqref{Orlando} of $\Omega$ tells us that $\Phi(x)=\sum_{n=-m}^{m} (1-|n|/x)^2$ because $1-|n|/x$ is nonpositive when $|n| \geq m+1$.
So $\Phi(x)=1+2 \sum_{n=1}^m (1-n/x)^2$ and we can employ the usual formulae for sums of consecutive integers and sums of consecutive squares to obtain the desired formula for $\Phi$.  Differentiation is then straightforward, as long as one remembers to check that the left- and right-hand derivatives of $\Phi(x)$ match when $x$ is a positive integer.

It is clear from our formulae for $\Phi^\prime(x)$ that $\Phi^\prime(x)=0$ for $x \in (0,1]$, and from the definition of $\Omega$ we see that $\Phi(x)=\sum_{n \in \Z} \max(0,1-|n|/x)^2$, which is a sum of nondecreasing functions.  The $n=1$ term of this sum is strictly increasing for $x > 1$, so that $\Phi^\prime(x) > 0$ when $x > 1$.
\end{proof}
\begin{lemma}\label{Nancy}
Let $\Psi\colon \Rplus \to \R$ with $\Psi(x)=\Omega(1/x,1/(2 x))$.  If $m$ is a positive integer, and $x \in [m-1/2,m+1/2]$, then
\begin{align*}
\Psi(x) & = 2 m - \frac{2 m^2}{x} + \frac{m(4 m^2-1)}{6 x^2} \\
\Psi^\prime(x) & = \frac{2 m^2}{x^2} - \frac{m(4 m^2-1)}{3 x^3}.
\end{align*}
We have $\Psi^\prime(x)=0$ for $x \in (0,1/2]$, while $\Psi^\prime(x) > 0$ for $x > 1/2$, so that $\Psi(x)=0$ for $x \in (0,1/2]$ and $\Psi(x) > 0$ for $x > 1/2$.
\end{lemma}
\begin{proof}
We write $2\Z+1$ for the set of all odd integers, and then the definition \eqref{Orlando} of $\Omega$ shows us that $\Psi(x)=\sum_{n \in 2\Z+1} \max(0,1-|n|/(2 x))^2$.  So if $x \in (0,1/2]$, we see that $\Psi(x)=0$.  

Let $m$ be a positive integer and suppose $x \in [m-1/2,m+1/2]$.  Then $\Psi(x)=\sums{n \in 2\Z+1 \\ |n| \leq 2 m-1} (1-|n|/(2 x))^2$ because $1-|n|/(2 x)$ is nonpositive when $|n| \geq 2 m+1$.
So $\Psi(x)=2 \sum_{k=1}^{m} (1-(2 k-1)/(2 x))^2$, and we employ can the usual formulae for sums of consecutive integers and sums of consecutive squares to obtain the desired formula for $\Psi$.  Differentiation is then straightforward, as long as one remembers to check that the left- and right-hand derivatives of $\Psi(x)$ match when $x$ is equal to half an odd positive integer.

It is clear from our formulae for $\Psi^\prime(x)$ that $\Psi^\prime(x)=0$ for $x \in (0,1/2]$.  Since $\Psi(x)=\sum_{n \in 2\Z+1} \max(0,1-|n|/(2 x))^2$ is the sum of nondecreasing functions and the $n=1$ term is strictly increasing when $x > 1/2$, we see that $\Psi^\prime(x)> 0$ when $x> 1/2$.
\end{proof}

\section*{Acknowledgements}
The author thanks Jonathan Jedwab and Amy Wiebe for helpful comments about this work.  The author thanks the anonymous referees for valuable corrections and suggestions that have improved this paper.

\end{document}